\documentclass[11pt]{article}

\usepackage{titlesec}
\usepackage[titletoc,toc,title]{appendix}

\usepackage[top=2cm, bottom=2cm, left=2.2cm, right=2.2cm]{geometry} %
\usepackage{amsmath}
\usepackage{amsfonts}
\usepackage{amsthm}
\usepackage{mathtools}
\usepackage{amssymb}
\usepackage[dvipsnames]{xcolor}
\usepackage[colorlinks = true,
            linkcolor = Blue,
            urlcolor  = OliveGreen,
            citecolor = BlueViolet,
            anchorcolor = blue]{hyperref}
\usepackage{complexity}
\usepackage{relsize}
\usepackage{graphicx, color}
\usepackage{epsfig}
\usepackage{caption}
\usepackage{subcaption}
\usepackage{multirow}
\usepackage{enumitem}
\usepackage{balance}
\usepackage{flushend}
\usepackage{epstopdf}
\usepackage{longtable}
\usepackage{stmaryrd}
\usepackage{pstricks}
\usepackage{pst-tree}
\usepackage{pst-all}
\usepackage{graphicx}
\usepackage{epstopdf}
\usepackage{lipsum}
\usepackage{xspace}
\usepackage{color}
\usepackage{framed}
\usepackage[normalem]{ulem}
\usepackage{float}
\usepackage{caption}
\usepackage{subcaption}
\usepackage{fancyvrb}
\usepackage{multirow}
\usepackage{bm}

\usepackage{algorithmic} 
\usepackage[ruled,linesnumbered,vlined]{algorithm2e}

\newtheorem{example}{Example}[section]

\newcommand{\random}{\xleftarrow{\$}}
\renewcommand{\class}[1]{{\ensuremath{\mathsf{#1}}}}

\newcommand{\query}{\ensuremath{\class{Query}}\xspace}
\newcommand{\enc}{\ensuremath{\class{Enc}}\xspace}
\newcommand{\Enc}{\ensuremath{\class{Enc}}\xspace}
\newcommand{\dec}{\ensuremath{\class{Dec}}\xspace}

\newcommand{\vect}[1]{\ensuremath{\textbf{#1}}}

\newcommand{\JC}{\cc{JC}}
\newcommand{\Score}{\cc{Score}}

\newcommand{\vx}{\vect{x}}
\newcommand{\vc}{\vect{c}}

\newcommand{\vq}{\vect{q}}
\newcommand{\valpha}{{\bm \alpha}}
\newcommand{\ignore}[1]{}

\newcommand{\vB}{\vect{B}}
\newcommand{\vE}{\vect{E}}

\newcommand{\F}[1]{\mathcal{F}_{\textbf{#1}}}
\renewcommand{\A}{\mathcal{A}}
\renewcommand{\L}{\mathcal{L}}

\renewcommand{\P}{\mathcal{P}}
\newcommand{\Z}{\mathbb{Z}}

\renewcommand{\cc}[1]{\class{#1}}

\newcommand{\setup}{\class{Setup}}
\newcommand{\token}{\class{Token}}

\renewcommand{\c}{\class{c}}

\newtheorem{theorem}{Theorem}[section]
\newtheorem{lemma}[theorem]{Lemma}

\newtheorem{definition}[theorem]{Definition}

\newtheorem{protocol}{Protocol}[section]

\newcommand{\bool}{\ensuremath{\{0,1\}}}

\newcommand{\larr}{\ensuremath{\leftarrow}}

\newcommand{\assign}{:=}

\newcommand{\angles}[1]{\langle #1 \rangle}

\renewcommand{\C}{\ensuremath{\mathcal{C}}}
\renewcommand{\E}{\ensuremath{\mathcal{E}}}

\renewcommand{\D}{\ensuremath{\textbf{D}}}

\renewcommand{\v}{\ensuremath{\mathbf{v}}}

\newcommand{\out}{\ensuremath{\mathsf{Output}}}

\newcommand{\sk}{\ensuremath{sk}}
\newcommand{\pk}{\ensuremath{pk}}

\newcommand{\cqa}{\ensuremath{{\sf CQA}}}

\newcommand{\tk}{\ensuremath{\mathsf{tk}}}

\renewcommand{\pk}{\ensuremath{{\sf pk}}}
\renewcommand{\sk}{\ensuremath{{\sf sk}}}

\newcommand{\ER}{\ensuremath{\mathsf{ER}}}

\newcommand{\EP}{\cc{EP}}

\newcommand{\Real}{\ensuremath{\mathbf{Real}}}
\newcommand{\Ideal}{\ensuremath{\mathbf{Ideal}}}

\renewenvironment{proof}{\textbf{Proof:}}{\hfill$\blacksquare$}

\newcommand{\EHL}{\cc{EHL}}
\newcommand{\eEHL}{\ensuremath{\cc{EHL}^+}}

\newcommand{\EncSort}{\cc{EncSort}}
\newcommand{\EncCompare}{\cc{EncCompare}}

\newcommand{\SecDedup}{\cc{SecDedup}}

\newcommand{\SecWorst}{\cc{SecWorst}}
\newcommand{\SecBest}{\cc{SecBest}}
\newcommand{\SecJoin}{\cc{SecJoin}}
\newcommand{\StripEnc}{\cc{StripEnc}}
\newcommand{\RecoverEnc}{\cc{RecoverEnc}}
\newcommand{\bitminus}{\ensuremath{\ominus}}

\newcommand{\En}[2][]{\enc_{#1}(#2)}
\newcommand{\Etwo}[1]{\ensuremath{\E^2\big(#1\big)}}

\newcommand{\hmac}{\ensuremath{\texttt{HMAC}}}

\renewenvironment{proof}{\emph{Proof:}}{\hfill$\blacksquare$}

\newlength{\protowidth} \protowidth \linewidth

\newcommand{\bottom}[1]{\underline{#1}}

\newcommand{\orand}{\ensuremath{\tilde{o}}}
\newcommand{\Wrand}{\ensuremath{\widetilde{W}}}
\newcommand{\Brand}{\ensuremath{\widetilde{B}}}
\newcommand{\Irand}{\ensuremath{\widetilde{I}}}
\newcommand{\Ihat}{\ensuremath{\hat{I}}}
\newcommand{\ohat}{\ensuremath{\hat{o}}}
\newcommand{\Bhat}{\ensuremath{\widehat{B}}}
\newcommand{\What}{\ensuremath{\widehat{W}}}

\newcommand{\secjoin}{\ensuremath{\bowtie_{\cc{sec}}}}

\newcommand{\sql}[1]{{\small \texttt{#1}}}

\newcommand{\SecFilter}{\cc{SecFilter}}
\newcommand{\Rand}{\cc{Rand}}
\newcommand{\srand}{\ensuremath{\tilde{s}}}
\newcommand{\Xrand}{\ensuremath{\widetilde{X}}}
\newcommand{\Trand}{\ensuremath{\widetilde{T}}}
\newcommand{\rrand}{\ensuremath{\widetilde{r}}}
\newcommand{\Rrand}{\ensuremath{\widetilde{R}}}
\newcommand{\sund}{\ensuremath{\overline{s}}}
\newcommand{\Xund}{\ensuremath{\overline{X}}}

\newcommand{\SIM}{\cc{Sim}}

\newcommand{\SecDupElim}{\cc{SecDupElim}}
\newcommand{\patients}{\texttt{patients}} 

\newcommand{\topk}{\cc{topk}}
\newcommand{\Fideal}{\ensuremath{\F{\topk}}}

\def\shownotes{1}

\mathchardef\mhyphen="2D
\ifnum\shownotes=1
\newcommand{\authnote}[2]{{\textcolor{red}{\textsf{#1 Notes: }\textcolor{blue}{ #2}}}}
\else
\newcommand{\authnote}[2]{}
\fi

\usepackage{authblk}

\SetKwProg{Def}{def}{:}{}

\newcommand{\SecTopK}{\cc{SecTopK}}
\newcommand{\Token}{\cc{Token}}
\newcommand{\SecQuery}{\cc{SecQuery}}

\SetKwBlock{ServerOne}{Server $S_1$:}{end}
\SetKwBlock{ServerTwo}{Server $S_2$:}{end}
\SetKwInput{SOne}{$S_1$'s input}
\SetKwInput{STwo}{$S_2$'s input}
\SetKwInput{SOneOut}{$S_1$'s ouput}

\begin{document}

\title{\bf Top-k Query Processing on Encrypted Databases with Strong Security Guarantees}
\author[1]{Xianrui Meng\thanks{Amazon AWS \texttt{xmeng@cs.bu.edu}} \,\,\,\, Haohan Zhu\thanks{ Facebook \texttt{zhu@cs.bu.edu}} \,\,\,\,  George Kollios\thanks{Boston University \texttt{gkollios@cs.bu.edu}}}

\date{}
\maketitle

\begin{abstract}
Privacy concerns in outsourced cloud databases have become more and 
more important recently and many efficient and scalable query processing
methods over encrypted data have been proposed.
However, there is very limited work on how to securely process top-$k$ 
ranking queries over encrypted databases in the cloud.
In this paper, we focus exactly on this problem: secure and efficient 
processing of top-$k$ queries over outsourced databases.
In particular, we propose the first efficient and provable secure 
top-$k$ query processing construction that achieves adaptively $\cqa$ security. 
We develop an encrypted data structure called {\EHL} and describe 
several secure sub-protocols under our security model to answer top-$k$ queries.
Furthermore, we optimize our query algorithms for both space and time efficiency.  
Finally, in the experiments, we empirically analyze our protocol using real world
datasets and demonstrate that our construction is efficient and practical.
\end{abstract}

\sloppy
\section{Introduction}
As  remote storage and cloud computing services emerge, such 
as Amazon's EC2, Google AppEngine, and Microsoft's Azure, many
enterprises, organizations, and end users may outsource their 
data to those cloud service providers for reliable maintenance, 
lower cost, and better performance. In fact, a number of 
database systems on the cloud have been developed recently 
that offer high availability and flexibility at relatively 
low costs. However, despite these benefits, there are still 
a number of reasons that make many users to  refrain from 
using these services, especially users with sensitive and 
valuable data. Undoubtedly, the main issue for this is related
to security and privacy concerns~\cite{cloudsecurity11}.
Indeed, data owner and clients may not fully trust a public 
cloud since some of hackers, or the cloud's administrators 
with root privilege can fully access all data for any purpose. 
Sometimes the cloud provider may sell its business to an 
untrusted company, which will have full access to the data.
One approach to address these issues is to encrypt the data
before outsourcing them to the cloud. For example, 
electronic health records (EHRs) should be encrypted before
outsourcing in compliance with regulations like 
HIPAA\footnote{\small HIPAA is the federal Health Insurance
Portability and Accountability Act of 1996.}. 
Encrypted data can bring an enhanced security into the 
Database-As-Service environment~\cite{HacigumusILM02}. 
However, it also introduces significant difficulties 
in querying and computing over these data.

Although top-$k$ queries are important query types in many database 
applications~\cite{IlyasBS08}, to the best of our knowledge, none of 
the existing works handle the top-$k$ queries securely and efficiently. 
Vaiyda et. al.~\cite{DBLP:conf/icde/VaidyaC05} studied privacy-preserving 
top-$k$ queries in which the data are vertically partitioned instead of 
encrypting the data. 
Wong et. al.~\cite{Wong2009} proposed an encryption scheme for knn queries 
and mentioned a method of transforming their scheme to solve top-$k$ 
queries, however, as shown in~\cite{DBLP:conf/icde/0002LX13}, their 
encryption scheme is not secure and is vulnerable to chosen plaintext attacks. 
Vaiyda et. al.~\cite{DBLP:conf/icde/VaidyaC05} also studied 
privacy-preserving top-$k$ queries in which the data are vertically 
partitioned instead of encrypting the data. 

We assume that the data owner and the clients are trusted, but not the cloud server. Therefore, the data owner encrypts each database relation $R$ using some probabilistic encryption scheme before outsourcing it to the cloud. An authorized user specifies a query $q$ and generates a \emph{token} to query the server. 
Our objective is to allow the cloud to securely compute the top-$k$ results based on a user-defined 
ranking function over $R$, and, more importantly,
the cloud should not learn anything about $R$ or  $q$.
Consider a real world example for a health medical database below:
\begin{example}
An authorized doctor, Alice, wants to get the top-$k$ results based on some
ranking criteria from the \emph{encrypted} electronic health record database $\patients$ (see
Table~\ref{tab:health_data}).
The encrypted $\patients$ database may contain several attributes;  
here we only list a few in Table~\ref{tab:health_data}: 
\emph{patient name, age, id number, 
trestbps~\footnote{trestbps: resting blood pressure (in mm Hg)}, 
chol\footnote{chol: serum cholestoral in mg/dl}, 
thalach\footnote{maximum heart rate achieved}}.

\begin{table}[tbph]
\centering
\begin{tabular}{|c|c|c|c|c|c|}\hline
patient name & age & id & trestbps & chol & thalach\\\hline\hline
$E(Bob)$ & $E(38)$ & $E(121)$ & $E(110)$ & $E(196)$ & $E(166)$ \\\hline
$E(Celvin)$ & $E(43)$ & $E(222)$ & $E(120)$ & $E(201)$ & $E(160)$\\\hline
$E(David)$ & $E(60)$ & $E(285)$ & $E(100)$ & $E(248)$ & $E(142)$\\\hline
$E(Emma)$ & $E(36)$ & $E(956)$ & $E(120)$ & $E(267)$ & $E(112)$\\\hline
$E(Flora)$ & $E(43)$ & $E(756)$ & $E(100)$ & $E(223)$ & $E(127)$\\\hline
\end{tabular}
\caption{Encrypted $\patients$ Heart-Disease Data}\label{tab:health_data}
\end{table}

One example of a top-$k$  query (in the form of a SQL query)  can be: 
\texttt{SELECT * FROM patients ORDERED BY chol+thalach STOP AFTER k}.
That is, the doctor wants to get the top-2 results based the score $chol+thalach$ from all the patient records.
However, since this table contains very sensitive information about the patients, 
the data owner first encrypts the table and then delegates it to the cloud.  
So, Alice requests a key from the data owner and generates a query \emph{token} based on the query.  
Then the cloud searches and computes on the {\it encrypted} table to find out the top-$k$ results. 
In this case, the top-2 results are the records of patients $David$ and $Emma$.
\end{example} 

Our protocol extends the No-Random-Access (NRA)~\cite{pods/FaginLN01} algorithm for computing top-$k$ queries over a probabilistically encrypted relational database. 
Moreover, our query processing model assumes that two non-colluding semi-honest clouds, 
which is the model that has been showed working well
(see~\cite{DBLP:conf/icde/ElmehdwiSJ14, DBLP:conf/cms/BugielNSS11, DBLP:conf/icde/LiuZLLZZ15,fc15/FO, DBLP:conf/ndss/BostPTG15}).
We encrypt the database in such a way that the server can obliviously execute NRA over the encrypted database without learning the underlying data. 
This is accomplished with the help of a secondary independent cloud server (or Crypto Cloud).
However, the encrypted database resides only in the primary cloud.
We adopt two efficient state-of-art secure protocols, $\EncSort$~\cite{fc15/FO} and 
$\EncCompare$~\cite{DBLP:conf/ndss/BostPTG15}, which are the two building block we need in our top-$k$ secure construction. We choose these two building blocks mainly because of their efficiency.

During the query processing, we propose several novel sub-routines that can securely
compute the best/worst score and de-duplicate replicated data items over the encrypted database.
Notice that our proposed sub-protocols can also be used as stand-alone building blocks for other applications as well.
We also would like to point out that during the querying phase the computation performed by the client is very small. The client only needs to compute a simple token for the server and 
all of the relatively heavier computations are performed by the cloud side.
Moreover, we also explore the problem of top-$k$ join queries over multiple encrypted relations.

We also design a \emph{secure top-k join operator}, denote as $\secjoin$,
to securely join the tables based on \emph{equi-join} condition.
The cloud homomorphically computes the top-$k$ join on the top of joined results and reports the encrypted top-$k$ results. 
Below we summarize our main contributions:
\begin{itemize}[noitemsep,nolistsep]
\item We propose a new practical protocol designed to answer top-$k$ 
      queries over encrypted relational databases. 
\item We propose an encrypted data structures called $\EHL$ which 
      allows the servers to homomorphically evaluate the equality relations between two objects.
\item We propose several independent sub-protocols such that the clouds can securely compute
      the best/worst scores and de-duplicate replicated encrypted objects with the use of another non-colluding server.
\item We also extend our techniques to answer top-$k$ join queries over multiple encrypted relations.
\item The scheme is experimentally evaluated using real-world datasets and result shows that our scheme is efficient and practical.
\end{itemize}

\section{Related Works and Background}
The problem of processing queries over the outsourced encrypted databases is not new. 
The work~\cite{HacigumusILM02} proposed executing SQL queries over encrypted data in the database-service-provider model using bucketization.
Since then, a number of works have appeared on executing various queries over encrypted data.

A significant amount of works have been done for privacy preserving keyword
search queries or boolean queries, 
such as~\cite{SWP00, DBLP:journals/jcs/CurtmolaGKO11, crypto/CashJJKRS13}. 
Recent work~\cite{DBLP:conf/esorics/SamanthulaJB14} proposed a general 
framework for boolean queries of disjunctive normal form queries on encrypted data.
In addition, many works have been proposed for range 
queries~\cite{DBLP:conf/sp/ShiBCSP07, DBLP:journals/vldb/HoreMCK12, LiLWB14}. 
Other relevant works include privacy-preserving data 
mining~\cite{DBLP:series/ads/2008-34, DBLP:journals/vldb/VaidyaKC08,  DBLP:journals/tkdd/VaidyaCKP08, LP00, DBLP:journals/tkde/KantarciogluC04}. 

Recent works in the cryptography community have shown that it is possible to perform arbitrary computations over encrypted data,
using fully homomorphic encryption (FHE)~\cite{G09b}, or Oblivious RAM~\cite{GO96}. 
However, the performance overheads of such constructions are very high in practice, thus they're not suitable for practical database queries.  
Some recent advancements in ORAM schemes~\cite{190906} show 
promise and can be potentially used in certain environments.
As mentioned, \cite{DBLP:conf/icde/VaidyaC05} is the only work that studied privacy preserving execution of top-$k$ queries. 
However, their approach is mainly based on the $k$-anonymity privacy policies, 
therefore, it cannot extended to encrypted databases.
Recently, differential privacy~\cite{DworkN04} has emerged as a powerful model
to protect against unknown adversaries with guaranteed probabilistic accuracy. 
However, here we consider {\em encrypted} data in the outsourced model;
moreover, we do not want our query answer to be perturbed by noise, 
but we want our query result to be exact.
Kuzu et. al.~\cite{Kuzu2014} proposed a scheme that leverages DP and leaks
obfuscated access statistics to enable efficient searching.
Another approach has been extensively studied is order-preserving 
encryption (OPE)~\cite{DBLP:series/ads/2008-34, AKSX04, DBLP:conf/crypto/BoldyrevaCO11,  DBLP:conf/sosp/PopaRZB11, LP00}, which preserves the order of the message. We note
that, by definition, OPE directly reveals the order of the objects' ranks, 
thus does not satisfy our data privacy guarantee. 
Furthermore, \cite{DBLP:conf/sigmod/HangKD15} proposed a prototype for access control using deterministic proxy encryption, 
and other secure database systems have been proposed by using embedded secure hardware, 
such as TrustedDB~\cite{DBLP:conf/sigmod/BajajS11} and Cipherbase~\cite{DBLP:conf/icde/ArasuEJKKR15}.

\paragraph{Secure $k$NN queries.}
One of the most relevant problems is answering $k$NN ($k$ Nearest Neighbor) queries. 
Note that top-$k$ queries should not be confused with similarity search, such as $k$NN queries.  
For $k$NN queries, one is interested in retrieving the $k$ most similar objects from the database to a query object, where the similarity between two objects is measured over some metric space, for instance using the $L_2$ metric.
Many works have been proposed to specifically handle $k$NN queries on encrypted data, such as~\cite{Wong2009,DBLP:conf/icde/0002LX13,DBLP:journals/tkde/ChoiGLB14}. 

A recent work~\cite{DBLP:conf/icde/ElmehdwiSJ14} proposed secure $k$NN query under the same architecture setting as ours. 
We would like to point out that their solution does not directly solve the problem of top-$k$ queries.
In particular, \cite{DBLP:conf/icde/ElmehdwiSJ14} designed a protocol for ranking distances between the query point and 
the records using the $L_2$ metric, while we consider the top-$k$ selection query based on a class of scoring functions using linear combinations of attribute values.
Nevertheless, if we follow the similar setup from \cite{DBLP:conf/icde/ElmehdwiSJ14}, we can define the scoring function to be the sum of the squares, i.e. $\sum x^2_i(o)$, where $x_i(o)$ is the $i$-th attribute value for object $o$ and is a positive value. 
Then one can adapt the secure $k$NN scheme by querying a large enough query point 
(say, the upper bound of the attribute value) to get the $k$-\emph{nearest-neighbors} and
therefore it can return top-$k$ results. We show in Section~\ref{experiment}, that even under this particular setting, our protocol is much more efficient than~\cite{DBLP:conf/icde/ElmehdwiSJ14}. The computational complexity, for each query, for~\cite{DBLP:conf/icde/ElmehdwiSJ14} is at least $O(nm)$, where $n$ is the number of records and $m$ the number of attributes. Furthermore, the communication overhead between the two clouds is also $O(nm)$. Thus, this protocol is not very efficient for even small sized  databases.

\section{Preliminaries}\label{section:topk-pre}
\subsection{Problem Definition}
Consider a data owner that has a database relation  $R$ of $n$ objects,
denoted by $o_1, \dots, o_n$, and each object $o_i$ has $M$ attributes. For simplicity,
we assume that all $M$ attributes take numerical values. 
Thus, the relation  $R$ is an $n \times M$ matrix. 
The data owner would like to outsource $R$ to a third-party cloud $S_1$ that is completely untrusted. Therefore,
data owner encrypts $R$ and sends the encrypted relation $\ER$ to the cloud. 
After that, any authorized client should be able to get the results of the top-$k$ query over this encrypted relation directly from $S_1$, by specifying $k$ and a score function over the $M$ (encrypted) attributes. 
We consider the monotone scoring (ranking) functions that are weighted linear combinations over all attributes, that is $F_W(o)$ $=$ $\sum w_i \times x_i(o)$, where each $w_i \geq 0$ is a user-specified weight for the $i$-th attribute and $x_i(o)$ is the local score (value) of the $i$-th attribute for object $o$. 
Note that we consider the monotone linear function mainly because it is the most important and widely used score function on top-$k$ queries~\cite{IlyasBS08}.
The results of a top-$k$ query are the objects with the highest $k$ scores of $F_W$ values.
For example, consider an authorized client, Alice, who wants to
run a top-$k$ query over the encrypted relation $\ER$. Consider the following
query:
\sql{q = SELECT * FROM ER ORDER BY } $F_W(\cdot)$ \sql{ STOP AFTER k;}
\noindent That is, Alice wants to get the top-$k$ results based on her scoring function $F_W$, 
for a specified set of weights.
Alice first has to request the keys from the data owner, then generates a \emph{query token} $\tk$. 
Alice sends the $\tk$ to the cloud server. The cloud server storing the 
encrypted database $\ER$ processes the top-$k$ query and sends the encrypted results back to Alice. In the real world scenarios, the authorized clients can locally store the keys for generating the token.

\subsection{The Architecture}
We consider the secure computation on the cloud under the \emph{semi-honest} 
(or \emph{honest-but-curious}) adversarial model. 
Furthermore, our model assumes the existence of two \emph{different} 
non-colluding semi-honest cloud providers, $S_1$ and $S_2$, where $S_1$ 
stores the encrypted database $\ER$ and $S_2$ holds the secret keys and 
provides the crypto services. We refer to the server $S_2$ as the {\it Crypto Cloud} 
and assume $S_2$ resides in the cloud environment and is isolated from $S_1$.
The two parties $S_1$ and $S_2$ do not  trust each other, and therefore, they have to execute secure computations on encrypted data.
The two parties $S_1$ and $S_2$ belong to two different cloud providers 
and do not trust each other; therefore, they have to execute secure 
computations on encrypted data. In fact, crypto clouds have been built 
and used in some industrial applications today
(e.g., the pCloud Crypto\footnote{\scriptsize\url{https://www.pcloud.com/encrypted-cloud-storage.html}} or boxcryptor\footnote{\scriptsize\url{https://www.boxcryptor.com/en/provider}}).
This model is not new and has already been widely used in related work, 
such as~\cite{DBLP:conf/icde/ElmehdwiSJ14, DBLP:conf/cms/BugielNSS11, DBLP:conf/icde/LiuZLLZZ15,fc15/FO, DBLP:conf/ndss/BostPTG15}. 
As pointed out by these works, we emphasize that these cloud services 
are typically provided by some large companies, such as Amazon, 
Microsoft Azure, and Google, who have also commercial interests not to collude.
The Crypto Cloud $S_2$ is equipped with a cryptographic processor, 
which stores the decryption key. 
The intuition behind such an assumption is as follows. Most of the cloud service providers in the market are well-established IT companies, such as Amazon AWS, Microsoft Azure and Google Cloud. Therefore, a collusion between them is highly unlikely as it will damage their reputation which effects their revenues. 
When the server $S_1$ receives the query token, $S_1$ initiates the secure computation protocol
with the Crypto Cloud $S_2$. Figure~\ref{fig:model} shows an overview of the architecture.

\begin{figure}[th!p]
\centering
    \includegraphics[width=0.5\textwidth]{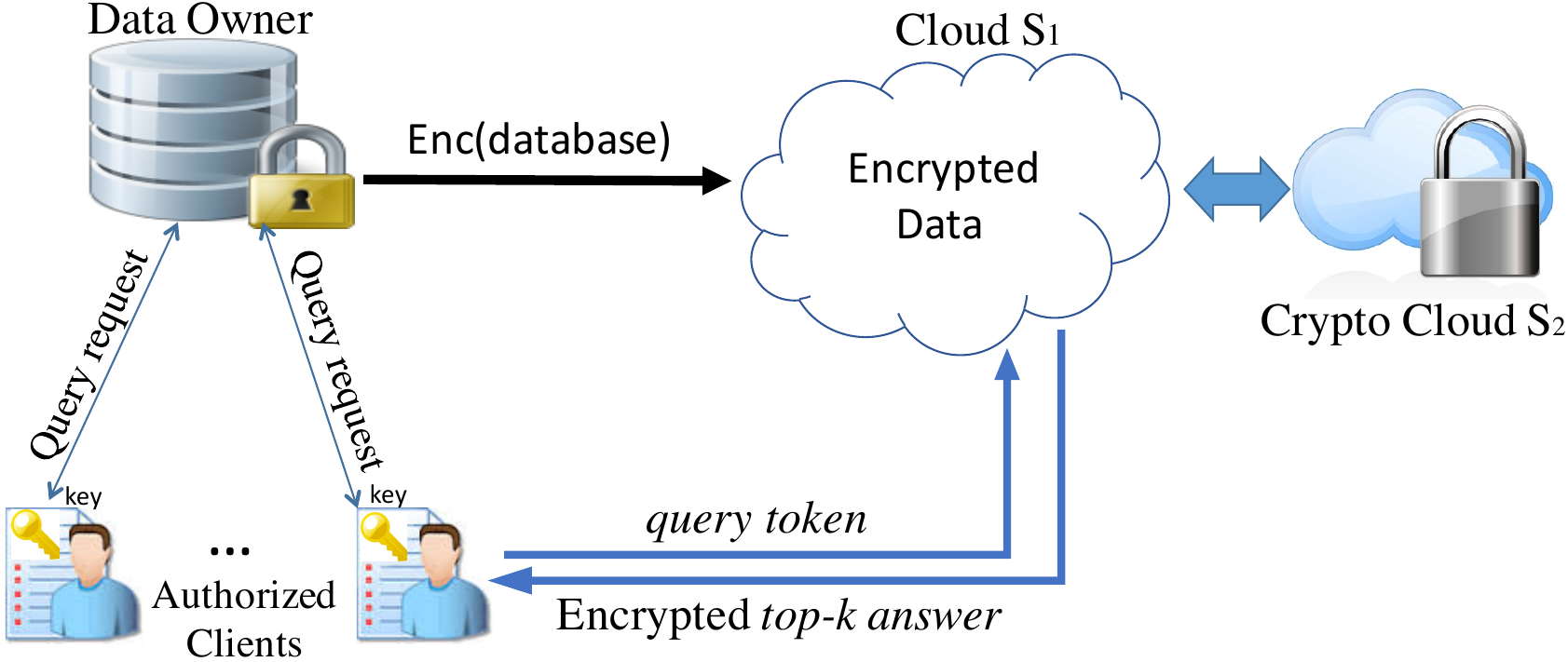}
    \caption{An overview of our model}
    \label{fig:model}
\end{figure}

\subsection{Cryptographic Tools}
In Table~\ref{table:notation} we summarize the notation. In the following, we present
the cryptographic primitives used in our construction.

\ignore{
}

\paragraph{Paillier Cryptosystem}
The Paillier cryptosystem~\cite{Paillier99} is a semantically secure 
public key encryption scheme.
The message space $\mathcal M$ for the encryption is $\Z_N$, where $N$ is a product of two large prime numbers $p$
and $q$. For a message $m\in \Z_N$, we denote $\En[\pk]{m} \in \Z_{N^2}$ to be the encryption of $m$ with the public key $\pk$. When the key is clear in the text, we simply use 
$\En{m}$ to denote the encryption of $m$ and $\dec_{\sk}(c)$ to denote the decryption of a
ciphertext $c$.
The details of encryption and decryption algorithm can be found in~\cite{Paillier99}.
It has the following homomorphic properties:
\begin{itemize}
\item{\emph{Addition:}} $\forall x, y \in \Z_N$, $\En{x}\cdot\En{y} = \En{x+y}$
\item{\emph{Scalar Multiplication:}} $\forall x, a \in \Z_N$, $\En{x}^a=\En{a\cdot x}$
\end{itemize}

\paragraph{Generalized Paillier} 
Our construction also relies on Damg{\aa}rd-Jurik(DJ) cryptosystem 
introduced by Damg{\aa}rd and Jurik \cite{DBLP:conf/pkc/DamgardJ01}, which is a generalization of Paillier encryption. 
The message space $\mathcal M$ expands to $\Z_{N^{s}}$ for $s\ge 1$, and the ciphertext space is under the group $\Z_{N^{s+1}}$. 
As mentioned in~\cite{DBLP:conf/tcc/AdidaW07}, this generalization allows one to doubly encrypt messages and use the additive homomorphism of the inner encryption layer under the same secret key.
In particular, let $\Etwo{x}$ denote an encryption of the DJ scheme for a message $x\in \Z_{N^2}$ (when $s = 2$) and $\En{x}$ be a normal Paillier encryption. 
This extension allows a ciphertext of the first layer to be treated as a plaintext in the second layer. 
Moreover, this nested encryption preserves the structure over inner ciphertexts and allows one to manipulate it as follows:
$$\Etwo{\En{m_1}}^{\En{m_2}} = \Etwo{\En{m_1}\cdot\En{m_2}} = \Etwo{\En{m_1+m_2}}$$
We note that this is the only homomorphic property that our construction relies on.

Throughout this paper, we use $\sim$ to denote that the underlying plaintext under encryption $\E$ are the same, i.e., $\Enc(x) \sim \Enc(y) \Rightarrow x = y$.
We summarize the notation throughout this paper in Table~\ref{table:notation}.
Note that in our application, we need one layered encryption; that is, given $\Etwo{\En{x}}$, we want a normal Paillier encryption $\En{x}$. As introduced in~\cite{fc15/FO}, this could simply be done 
with the help of $S_2$. However, we need a protocol $\RecoverEnc$ to securely remove one layer of encryption.

\begin{table}[tpb]
\centering
\begin{tabular}{| c | c |} \hline
\textbf{Notation} & \textbf{Definition} \\ \hline
$n$ & Size of the relation $R$, i.e. $|R|= n$ \\ \hline
$M$ & Total number of attributes in $R$ \\ \hline
$m$ & Total number of attributes for the query $q$	 \\ \hline
$\En{m}$ & Paillier encryption of $m$ \\ \hline
$\dec(c)$  & Paillier decryption of $c$ \\ \hline
$\Etwo{m}$ &  Damg{\aa}rd-Jurik (DJ) encryption of $m$ \\ \hline
$\enc(x) \sim \enc(y)$ & Denotes $x = y$, i.e. $\dec(\Enc(x)) = \dec(\Enc(y))$ \\ \hline
$\EHL(o)$ & Encrypted Hash List of the object $o$\\ \hline
$\eEHL(o)$ & Efficient Encrypted Hash List of the object $o$\\ \hline
$\bitminus$, $\odot$ & $\EHL$ and $\eEHL$ operations, see Section~\ref{sec:ehl}.\\ \hline
$I_i^d$ & The data item in the $i$th sorted list $L_i$ at depth $d$ \\ \hline
$E(I_i^d)$ & Encrypted data item $I_i^d$ \\ \hline
$F_W(o)$ & Cost function in the query token\\ \hline
$B^d(o)$ & The best score (upper bound) of $o$ at depth $d$\\ \hline
$W^d(o)$ & The worst score (lower bound) of $o$ at depth $d$\\ \hline
\end{tabular}
\caption{Notation Summarization}
\label{table:notation}
\end{table}

\subsection{No-Random-Access (NRA) Algorithm}

\begin{algorithm}[th!]
\Def{NRA \big($L_1, ..., L_M$\big)}{
Do sorted access in parallel to each of the $M$ sorted lists $L_i$. 
At each depth $d$:
\Repeat{}{
 Maintain the bottom values $\bottom{x}_1^d, \bottom{x}_2^d, ..., \bottom{x}_M^d$ encountered in the lists\;
 For every object $o_i$ compute a lower bound $W^d(o_i)$ and upper bound $B^d(o_i)$\;
 Let $T^d_k$, the current top $k$ list, contain the $k$ objects with the largest $W^d(\cdot)$ values seen so far (and their grades), and let $M^d_k$ be the $k$th largest lower bound value, $W^d(\cdot)$ in $T_k^d$\;
 Halt and return $T^d_k$ when at least $k$ distinct objects have been seen (so that in particular $T^d_k$ contains $k$ objects) and when $B^d(o_k)\le M^d_k$ for all $o_k\notin T_k^d$, 
 i.e. the upper bound for every object who's not in $T_k^d$ is no greater than $M^d_k$. Otherwise, go to next depth\;
}
}
\caption{NRA Algorithm~\cite{pods/FaginLN01}} \label{alg:nra}
\end{algorithm}

The NRA algorithm~\cite{pods/FaginLN01} finds the top-$k$ answers by exploiting only sorted accesses to the relation $R$. 
The input to the NRA algorithm is a set of sorted lists $S$, each ranks
the ``same'' set of objects based on different attributes. The output is a ranked list of these
objects ordered on the aggregate input scores.
We opted to use this algorithm because it provides a scheme that 
leaks minimal information to the cloud server (since during query processing there is no need to access intermediate objects).
We assume that each column (attribute) is sorted independently to create a set of sorted lists $S$. 
The set of sorted lists is equivalent to the original relation, but the objects in each list $L$
are sorted in ascending order according to their local score (attribute value).
After sorting, $R$ contains $M$ sorted lists, denoted as $S = \{ L_1, L_2, \dots, L_M \}$.  
Each sorted list consists of $n$ data items, denoted as $L_i = \{ I_i^1, I_i^2, \dots, I_i^n \}$.
Each data item is a object/value pair $I_i^d = (o_i^d, x_i^d)$, where $o_i^d$ 
and $x_i^d$ are the object id and local score at the depth $d$ (when $d$ objects have been accessed under sorted access in each list) 
in the $i$th sorted list respectively. 
Since it produces the top-$k$ answers using bounds computed over their exact scores, NRA may not report the exact object scores. 
The score lower bound of some object $o$, $W(o)$, is obtained by applying the ranking function on $o$'s known scores 
and the minimum possible values of $o$'s unknown scores. The score upper bound of $o$, $B(o)$, 
is obtained by applying the ranking function on $o$'s known scores
and the maximum possible values of $o$'s unknown scores, which are the same as the last seen scores in the
corresponding ranked lists.
The algorithm reports a top-$k$ object even if its score is not precisely known. 
Specifically, if the score lower bound of an object $o$ is not below the score upper bounds of all other objects (including unseen objects), then $o$ can be safely reported as the next top-$k$ object. We give the details of the NRA in Algorithm~\ref{alg:nra}.

\section{Scheme Overview}
In this section, we give an overview of our scheme. 
The two non-colluding semi-honest cloud servers are denoted by $S_1$ and $S_2$. 

\begin{definition}
Let $\SecTopK = (\Enc, \Token, \SecQuery)$ be the secure top-$k$ query scheme containing three algorithms $\Enc$, $\Token$ and $\SecQuery$.
\begin{itemize}
	\item $\Enc(\lambda, R)$:
	 is the probabilistic encryption algorithm that takes relation $R$ and security parameter 
	 $\lambda$ as its inputs and outputs the encrypted relation $\ER$ and secret key $K$. 
	\item $\Token(K, q)$: takes a query $q$ and secret key $K$. It outputs a token $\tk$ for the query $q$.
	\item $\SecQuery(\tk, \ER)$ is the query processing algorithm that takes the token $\tk$ and $\ER$
	 and securely computes top-$k$ results based on the $\tk$.
\end{itemize}
\end{definition}
As mentioned earlier, our encryption scheme takes advantage of the NRA top-$k$ algorithm. 
The idea of $\Enc$ is to encrypt and permute the set of sorted lists 
for $R$, so that the server can execute a variation of the NRA algorithm 
using only sequential accesses to the encrypted data. 
To do this encryption, we design a new encrypted data structure for the objects, called $\EHL$.
The $\Token$ computes a token that serves as a trapdoor so that the cloud knows which list to access. 
In $\SecQuery$, $S_1$ scans the encrypted data depth by depth for each targeted list, 
maintaining a list of encrypted top-$k$ object ids per depth until there are $k$ 
encrypted object ids that satisfy the NRA halting condition.
During this process, $S_1$ and $S_2$ learn nothing about the underlying scores and objects.
At the end of the protocol, the object ids can be reported to the client. 
As we discuss next, there are two options after that.
Either the encrypted records are retrieved and returned to the client, or the client retrieves the records using oblivious RAM~\cite{GO96} that does not even reveal the location of the actual encrypted records. In the first case, the server can get some additional information by observing the access patterns, i.e., the encrypted results of different 
queries. However, there are schemes that address this access leakage~\cite{DBLP:conf/ndss/IslamKK12, Kuzu2014} and is beyond the scope of this paper. The second approach may be more expensive but is completely secure.

In the following sections, we first discuss the new encrypted data structures $\EHL$ and $\eEHL$. 
Then, we present the three algorithms $\Enc$, $\Token$ and $\SecQuery$ in more details.

\section{Encrypted Hash List (EHL)} \label{sec:ehl}
In this paper, we propose a new data structure called 
encrypted hash list ($\EHL$) to encrypt each object.  
The main purpose of this structure is to allow the cloud to 
homomorphically compute equality between the objects, 
whereas it is computationally hard for the server to figure out what the objects are.
Intuitively, the idea is that given an object $o$ we use $s$ 
Pseudo-Random Function (PRF) to hash the object into a binary list 
of length $H$ and then encrypt all the bits in the list to generate $\EHL$.  
In partilar, we use the secure key-hash functions $\hmac$ as the PRFs.
Let $\EHL(o)$ be the encrypted list of an object $o$ and let $\EHL(o)[i]$ 
denote the $i$th encryption in the list.
In particular, we initialize an empty list $\EHL$ of length $H$ and fill all the entries with $0$. 
First, we generate $s$ secure keys $\kappa_1, ..., \kappa_s$. The object $o$ is hashed to a list as follows:
1) Set $\EHL[\hmac(\kappa_i, o)\mod H] = 1$ for $1\le i \le s$. 
2) Encrypt each bit using Paillier encryption: for $0 \le j \le H-1$, $\Enc(\EHL(o)[j])$.
Fig.~\ref{fig:ehl} shows how we obtain $\EHL(o)$ for the object $o$.

\begin{figure}[th!bp]
   \centering
      \includegraphics[width=0.4\textwidth]{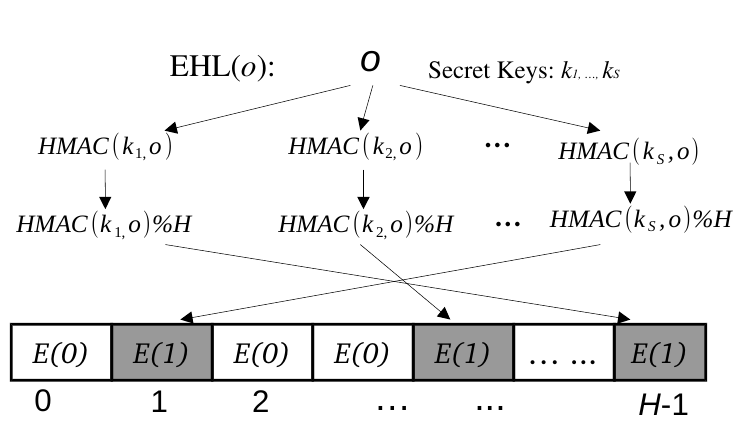}
   \caption{Encrypted Hash List for the object $o$.}\label{fig:ehl}
\end{figure}

\begin{lemma}\label{lemma:ehl-indis}
Given two objects $o_1$ and $o_2$, their $\EHL(o_1)$ and $\EHL(o_2)$ are computationally 
indistinguishable.
\end{lemma}
It is obvious to see that Lemma~\ref{lemma:ehl-indis} holds since the bits in the $\EHL$ are encrypted by the semantically secure Paillier encryption scheme.  
Given $\EHL(x)$ and $\EHL(y)$, we define the \emph{randomized operation} $\bitminus$ between 
$\EHL(x)$ and $\EHL(y)$ as follows: 
\begin{align}\label{eq:ehl-op} 
\EHL(x) \bitminus \EHL(y) \overset{def}{=} \prod^{H-1}_{i = 0} \big(\EHL(x)[i] \cdot \EHL(y)[i]^{-1}\big)^{r_{i}}
\end{align}
where each $r_{i}$ is some random value in $\Z_N$.

\begin{lemma}\label{lemma:eqi-bit}
Let $\En{b}$ $=$ $\EHL(x) \bitminus \EHL(y)$. Then the plaintext $b = 0$ if $x = y$ (two objects are the same),
otherwise $b$ is uniformly distributed in the group $\Z_N$ with high probability.
\end{lemma}

\begin{proof}
Let $\En{x_i} = \EHL(x)[i]$ and $\En{y_i} = \EHL(y)[i]$. 
If $x = y$, i.e. they are the same objects, then for all $i\in[0, H-1]$, 
$x_i = y_i$. Therefore, 
\[
\prod_{i = 0}^{H-1} (\EHL(x)[i] \cdot \EHL(y)[i]^{-1})^{r_i} 
= \E\big(\sum_{i=0}^{H-1}(r_i(x_i - y_i))\big) = \En{0}
\]
In the case of $x \ne y$, it must be true, with high probability, that there exists some 
$i \in [0, H-1]$ such that $\En{x_i}\nsim\En{y_i}$, i.e. the underlying 
bit at location $i$ in $\EHL(x)$ is different from the bit in $\EHL(y)$.
Suppose $\EHL(x)[i] = \En{1}$ and $\EHL(y)[i] = \En{0}$.
Therefore, the following holds: 
$$\left(\EHL(x)[i]\cdot\EHL(y)[i]^{-1}\right)^{r_i} = \En{r_i(1-0)} = \En{r_i}$$
Hence, based on the definition $\bitminus$, it follows that $b$ becomes random value uniformly 
distributed in the group $\Z_N$.
\end{proof}

It is worth noting that one can also use BGN cryptosystem for the similar operations above, as the BGN scheme can homomorphically evaluate quadratic functions.

\paragraph{False Positive Rate.}
Note that the construction is indeed a probabilistically encrypted Bloom Filter except that we
use one list for each object and encrypt each bit in the list.
The construction of $\EHL$ may report some false positive results for its $\bitminus$ operation, i.e. $\En{0}$ $\larr$ $\EHL(x)\bitminus\EHL(y)$ when $x \ne y$.  
This is due to the fact that $x$ and $y$ may be hashed to exactly the same
locations using $s$ many $\hmac$s. 
Therefore, it is easy to see that the false positive rate (FPR) is the same as the FPR of the Bloom Filter, where we can choose the number of hash functions $\hmac$ $s$ to be $\frac{H}{n}\ln2$ to minimize the false positive rate to be $(1 - (1 - {\frac{1}{H}}^{sn}))^s
\approx (1-e^{-sn/H})^s \approx 0.62^{H/n}$.
To reduce the false positive rate, we can increase the length of the list $H$.  
However, this will increase the cost of the structure both in terms of space overhead and number of operations for the randomization operation which is $O(H)$.
In the next subsection, we introduce a more compact and space-efficient encrypted data structure $\eEHL$.

\paragraph{EHL$^+$.}
We now present a computation- and space-efficient encrypted hash list $\eEHL$.
The idea of the efficient $\eEHL$ is to first `securely hash' the object $o$ to a larger space
$s$ times and only encrypt those hash values. 
Therefore, for the operation $\bitminus$, we only homomorphically subtract those hashed values. 
The complexity now reduces to $O(s)$ as opposed to $O(H)$, where $s$ is the number of the
secure hash functions used. We show that one can get negligible false positive rate even using a very small $s$.
To create an $\eEHL(o)$ for an object $o$, we first generate $s$ secure keys $k_1, ..., k_s$, then initialize a list $\eEHL$ of size $s$. 
We first compute $o_i \larr \hmac(k_i, o) \mod N$ for $1\le i \le s$. 
This step maps $o$ to an element in the group $\Z_N$, i.e. the message space for Paillier encryption.
Then set $\eEHL[i] \larr \En{o_i}$ for $1\le i \le s$.
The operation $\bitminus$ between $\eEHL(x)$ and $\eEHL(y)$ are similar defined as
in Equation(\ref{eq:ehl-op}), i.e. $\eEHL(x) \bitminus \eEHL(y)$ $\overset{def}{=}$ 
$\prod^{s-1}_{i = 0} \big(\EHL(x)[i] \cdot \EHL(y)[i]^{-1}\big)^{r_{i}}$, where each $r_i$ 
is some randomly generated value in $\Z_N$.
Similarly, $\eEHL$ has the same properties as $\EHL$. Let $\En{b}\larr\eEHL(x)\bitminus\eEHL(y)$,
$b = 0$ if $x = y$ and otherwise $b$ is random in $\Z_N$ with high probability. 

We now analyze the false positive rate (FPR) for $\eEHL$. 
The false positive answer occurs when $x \ne y$ and $\En{0}\larr \eEHL(x)\bitminus\eEHL(y)$. 
That is $\hmac(k_i, x)\% N = \hmac(k_i,y)\%N$ for each $i\in[1, s]$.
Assuming $\hmac$ is a Pseudo-Random Function, the probability of this happens is at most 
$\frac{1}{N^s}$. 
Taking the union bound gives that the FPR is at most ${n \choose 2}\frac{1}{N^s} \le \frac{n^2}{N^s}$. 
Notice that $N \approx 2^\lambda$ is large number as $N$ is the product of two large primes 
$p$ and $q$ in the Paillier encryption and $\lambda$ is the security parameter. 
For instance, if we set $N$ to be a $256$ bit number ($128$-bit primes in Paillier) 
and set $s = 4$ or $5$, then the FPR is negligible even for millions of records.
In addition, the size of the $\eEHL$ is much smaller than $\EHL$ as it stores only $s$ encryptions.
In the following section, we simply say $\EHL$ to denote the encrypted hash list using the $\eEHL$ structure.

\paragraph{Notation.}
We introduce some notation that we use in our construction. 
Let $\vx = (x_1, \dots, x_s)\in\Z_N^s$ and let the encryption $\En{\vx}$ denotes
the concatenation of the encryptions $\En{x_1}...\En{x_s}$.  
 Also, we denote by $\odot$ the block-wise multiplication between $\En{\vx}$ and $\EHL(y)$; 
that is, $\vc\larr\En{\vx}\odot\EHL(y)$, where $\vc_i \larr \En{x_i}\cdot\EHL(y)[i]$ for $i\in[1, s]$.

\section{Database Encryption}\label{sec:enc-setup}
We describe the database encryption procedure $\Enc$ in this section.
Given a relation $R$ with $M$ attributes, the data owner first encrypts the
relation using Algorithm~\ref{alg:encryption}.

\begin{algorithm}[th!]
 Given the relation $R$, sort each $L_i$ based on the attribute's value for $1\le i \le M$\;
 Generate a public/secret key $\pk_\p, \sk_\p$ for the Paillier encryption scheme and             random secret keys $\kappa_1, \dots, \kappa_s$ for $\EHL$\;
 Do sorted access in parallel to each of the $M$ sorted lists $L_i$\;
 \ForEach{data item $I_i = \angles{o^{d}_i, x^{d}_i} \in L_i$}{
    \ForEach{depth $d$}{
     Compute $\EHL(o^{d}_i)$ using the keys $\kappa_1, \dots, \kappa_s$\;
     Compute $\enc_{\pk_\p}({x}^d_{i})$ using $\pk_\p$\;
     Store the item $E(I_i^d) = \angles{\EHL(o^d_i), \enc_{\pk_\p}(x^d_i})$ at depth $d$\;
    }
  }
 Generate a secret key $K$ for a pseudorandom permutation $P$ and permute all the list based on $g$. For $1 \le i \le M$, permute $L_i$ as $L_{P_K(i)}$\;
 The data owner securely uploads the keys $\pk_\p, \sk_\p$ to the $S_2$, and only $\pk_\p$ to $S_1$\;
 Finally, each permuted list contains a list of encrypted item of the form $E(I^d) = \angles{\EHL(o^d), \En[\pk_\p]{x^d}}$. 
 Output all lists of encrypted items as the encrypted relation as $\ER$\;
 \caption{$\enc(R)$: Relation encryption} \label{alg:encryption}
\end{algorithm}

In $\ER$ each data item $I^{d}_i = (o^{d}_{i}, x^{d}_{i})$ at depth $d$ in the
sorted list $L_i$ is encrypted as $E(I^{d}_i) = \angles{\EHL(o^{d}_{i}),\enc_{\pk_\p}(x^{d}_{i})}$. 
As all the score has been encrypted under the public key $\pk_\p$, for the rest of the paper, we simply use $\enc(x)$ to denote the encryption $\enc_{\pk_\p}(x)$ under the public key $\pk_\p$.
Besides the size of the database and $M$, the encrypted $\ER$ doesn't reveal anything. In Theorem~\ref{thm:indis},
we demonstrate this by showing that two encrypted databases are indistinguishable if they have the same size and number of attributes. We denote $|R|$ by the size of a relation $R$.
\begin{theorem}\label{thm:indis}
Given two relations $R_1$ and $R_2$ with $|R_1| = |R_2|$ and same number of attributes. The encrypted $\ER_1$ and $\ER_2$ output by the algorithm $\Enc$ are indistinguishable.
\end{theorem}
The proof is straight forward as it's easy to see that the theorem holds based on Lemma~\ref{lemma:ehl-indis} and Paillier encryption scheme.


\section{Query Token}\label{sec:topk-token}
Consider the SQL-like query 
{\small\sql{q = SELECT * FROM ER ORDERED BY} $F_W(\cdot)$ \sql{STOP BY k}}, 
where $F_W(\cdot)$ is a weighted linear combination of all attributes. 
In this paper, to simplify our presentation of the protocol, we consider binary weights and therefore the scoring function is just a sum of the values of a subset of attributes.
However, notice that for non $\bool$ weights the client should provide these weights to the server and the server can simply adapt the same techniques by using the scalar multiplication property of the Paillier encryption before it performs the rest of the protocol which we discuss next.
On input the key $K$ and query $q$, the $\Token(K, q)$ algorithm is quite simple and works as follows:
the client specifies the scoring attribute set $\M$ of size $m$, i.e. $|\M| = m\le M$, then requests the key $\K$ from the data owner, where $\K$ is the key corresponds the Pseudo Random Permutation $P$.
Then the client computes the ${P_\K(i)}$ for each $i\in \M$ and sends
the following query token to the cloud server $S_1$:
$\tk = \sql{SELECT * FROM \ER \ ORDERED BY}$ $\{P_\K(i)\}_{i\in \M} \sql{ STOP BY k}$.

\newcommand{\SecUpdate}{\cc{SecUpdate}}

\section{Top-k Query Processing}\label{querysection}
As mentioned, our query processing protocol is based on the NRA algorithm. 
However, the technical difficulty is to execute the algorithm on the encrypted 
data while $S_1$ does not learn any object id or any score and attribute value of the data.
We incorporate several cryptographic protocols to achieve this. Our 
query processing uses two state-of-the-art efficient and secure protocols: 
$\EncSort$ introduced by~\cite{fc15/FO} and $\EncCompare$ introduced by~\cite{DBLP:conf/ndss/BostPTG15}
as building blocks. 
We skip the detailed description of these two protocols since they are not the focus of this 
paper. Here we only describe their functionalities:
1). $\EncSort$: $S_1$ has a list of encrypted keyed-value pairs
$(\En{key_1},\En{a_1})...(\En{key_m},\En{a_m})$ and a public key $\pk$, 
and $S_2$ has the secret key $sk$. At the end of the protocol, $S_1$ obtains a list 
\emph{new} encryptions $(\En{key'_1},\En{a'_1})...(\En{key'_m},\En{a'_m})$, 
where the key/value list is sorted based on the order $a'_1 \le a'_2 ...\le a'_m$ 
and the set $\{(key_1, a_1), ..., (key_m, a_m)\}$ is the same as 
$\{(key'_1,a'_1), ..., (key'_m, a'_m)\}$.
2). $\EncCompare(\En{a}, \En{b})$: $S_1$ has a public key $\pk$ and two encrypted values 
$\En{a}, \En{b}$, while $S_2$ has the secret key $\sk$.
At the end of the protocol, $S_1$ obtains the bit $f$ such that $f \assign (a\le b)$.
Several protocols have been proposed for the functionality above.
We choose the one from~\cite{DBLP:conf/ndss/BostPTG15} mainly because it is efficient and perfectly suits our requirements.

\subsection{Query Processing: $\SecQuery$}

We first give the overall description of the top-$k$ query processing $\SecQuery$ at a high level. 
Then in Section~\ref{sec:building-blocks}, we describe in details the secure sub-routines that we use in the query processing: $\SecWorst$, $\SecBest$, $\SecDedup$, and $\SecUpdate$.

As mentioned, $\SecQuery$ makes use of the NRA algorithm but is different from the original NRA, 
because $\SecQuery$ cannot maintain the global worst/best scores in plaintext. 
Instead, $\SecQuery$ has to run secure protocols depth by depth and homomorphically 
compute the worst/best scores based on the items at each depth. 
It then has to update the complete list of encrypted items seen so far with their global worst/best scores.
At the end, server $S_1$ reports $k$ encrypted objects (or object ids) without learning any object or its scores.

\paragraph{Notations.}
In the encrypted database, we denote each \emph{encrypted item} by $E(I) = \angles{\EHL(o), \enc(x)}$, 
where $I$ is the item with object id $o$ and score $x$.
During the query processing, the server $S_1$ needs to maintain the encrypted item with its current best/worst scores, 
and we denote by $\vE(I) = (\EHL(o), \enc(W), \enc(B))$ the \emph{encrypted score item} $I$ with object id $o$ with best score $B$ and worst score $W$.

\begin{algorithm}[th!]
	$S_1$ receives $\token$ from the client\;
    Parses the $\token$ and let $L_i = L_{P_K(j)}$ for $j\in \M$\;
	\ForEach{depth $d$ at each list}{
 		\ForEach{$E(I_i^d) = \angles{\EHL(o_i^d), \enc(x_i^d)}$ $\in$ $L_i$}{
 		    \tcc{Compute the worst score for object $o_i^d$ at current depth $d$}
		    Compute $\En{W^{d}_i)}\larr\SecWorst(E(I_i^{d}), H, \pk_\p, \sk_\p)$,
			where $H$ = $\{E(I_j^{d})\}_{j\in m, i \ne j}$\;
			\tcc{Compute the best score for object $o_i^d$ at current depth $d$}
     		Compute $\En{B^{d}_i}\larr\SecBest(E(I_i^d),\{j\}_{j\ne i}, \pk_\p, \sk_\p)$\;
     	 }
     	 \tcc{gets encrypted list $\Gamma_d$ without duplicated objects}
		 Run $\Gamma_d\larr\SecDedup(\{\vE(I_i^d)\}, \pk_\p,\sk_\p)$ with $S_2$ 
		 and get the local encrypted list $\Gamma_d$\label{line:topk-dedup}\;
  		Run $T^d\larr\SecUpdate(T^{d-1}, \Gamma_d, \pk_\p, \sk_\p)$ with $S_2$ and get $T^d$\label{line:dedup}\;
		If $|T^d| < k$ elements, go to the next depth.
        Otherwise, run $\EncSort(T^d)$ by sorting on $\En{W_i}$, 
        get first $k$ items as $T^d_k$\label{line:enc-sort}\;
  		Let the $k$th and the $(k+1)$th item be $E(I'_k)$ and $E(I'_{k+1})$, 
			$S_1$ then runs $f\larr\EncCompare(E(W'_k), E(B'_{k+1}))$ 
			with $S_2$, where $E(W'_k)$ is the worst score for $E(I'_k)$, 
            and $E(B'_{k+1})$ is the best score for $E(I'_{k+1})$ in $T^d$\;
	  	\If{$f = 0$}{{\bf Halt} and return the encrypted first
			 $k$ item in $T^d_k$}\label{line:halt}
	}
\caption{Top-$k$ Query Processing: $\SecQuery$}\label{alg:topk}
\end{algorithm}

In particular, upon receiving the token $\tk=$ \sql{SELECT * FROM ER ORDERED BY} $\{P_K(i)\}_{i\in \M}$ \sql{ STOP BY k}, 
the cloud server $S_1$ begins to process the query. 
The token $\tk$ contains $\{P_K(i)\}_{i\in \M}$ which informs $S_1$ to perform the sequential access to the lists $\{L_{P_K(i)}\}_{i\in \M}$.
By maintaining an encrypted list $T$, which includes items with their encrypted global best and worst scores, $S_1$ updates the list $T$ depth by depth. 
Let $T^d$ be the state of the encrypted list $T$ after depth $d$. 
At depth $d$, $S_1$ first homomorphically computes the local encrypted worst/best scores for each item appearing at this depth by running $\SecWorst$ and $\SecBest$.

In $\SecWorst$, $S_1$ takes the input of the current encrypted item $E(I^d_i) = \angles{\EHL(o_i^d), \enc(x_i^d)}$ and all of the encrypted items in other lists $H$ at current depth, i.e., $ H = \{E(I^d_j)\}_{j \ne i, j\in \M}$.
$S_1$ runs the protocol $\SecWorst$ with $S_2$, and obtains the encrypted worst score for the object $o^d_i$.
Similarly, in the protocol $\SecBest$, $S_1$ takes the input of the current encrypted item $E(I^d_i) = \angles{\EHL(o_i^d), \enc(x_i^d)}$ 
and the list pointers $\{j\}_{j\ne i}$ that indicates all of the encrypted item seen so far. $S_1$ runs the protocol $\SecBest$ with $S_2$, and obtains the encrypted worst score for the object $o^d_i$.
Then $S_1$ securely replaces the duplicated encrypted objects with large encrypted worst scores $Z$ by running $\SecDedup$ with $S_2$. 
In the $\SecDedup$ protocol, $S_1$ inputs the current encrypted items, $\{E(I_i^d)\}$, seen so far. 
After the execution of the protocol, $S_1$ gets list of encrypted items $\Gamma_d$ such that there are no duplicated objects.
Next, $S_1$ updates the encrypted global list from state $T^{d-1}$ to state $T^{d}$ by applying $\SecUpdate$.
After that, $S_1$ utilizes $\EncSort$ to sort the distinct encrypted objects with their scores in $T^d$ to obtain 
the first $k$ encrypted objects which are essentially the top-$k$ objects based on their worst scores so far.
The protocol halts if at some depth, the encrypted best score of the $(k$$+$$1)$-th object, $\En{B_{k+1}}$,
is less than the $k$-th object's encrypted worst score $\En{W_k}$. This can be checked
by calling the protocol $\EncCompare(\En{W_k}, \En{B_{k+1}})$.
Followed by underlying NRA algorithm, it is easy to see that $S_1$ can correctly reports  the encrypted top-$k$ objects.  
We describe the detailed query processing in Algorithm~\ref{alg:topk}.

\subsection{Building Blocks}\label{sec:building-blocks}
In this section, we present the detailed description of the protocols $\SecWorst$, $\SecBest$, $\SecDedup$, 
and $\SecUpdate$.

\begin{figure}[tbh!]
\centering
   \begin{minipage}[b]{.6\textwidth}
   	
  	  \begin{subfigure}[t]{\textwidth} 
           \includegraphics[width=\textwidth]{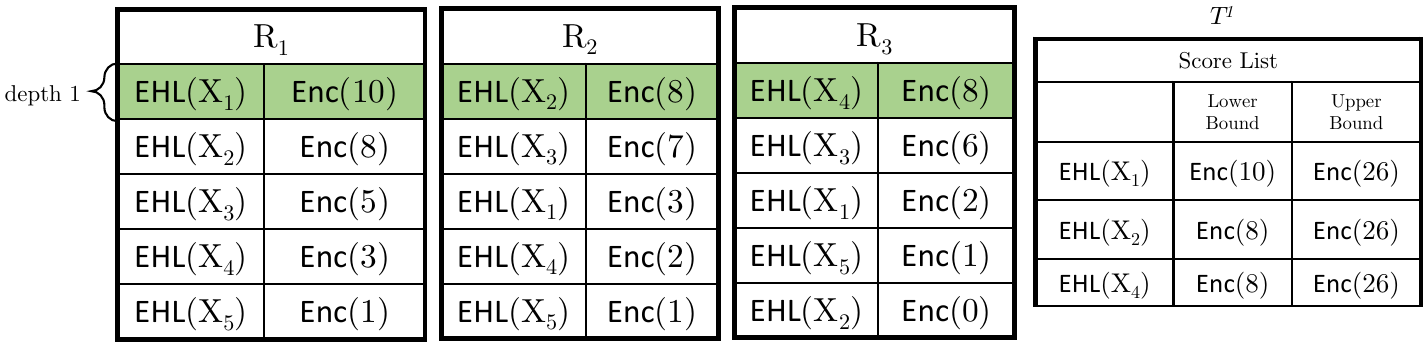}
           \caption{\centering{$\SecWorst$ $\&$ $\SecBest$ at depth 1.
           $T^1$ maintains the encrypted scores after depth 1.}}\label{fig:depth1}
       \end{subfigure}%
   \end{minipage}
   ~
   \begin{minipage}[b]{.6\textwidth}
   		\centering
    	\begin{subfigure}[t]{\textwidth}
          \centering  
           \includegraphics[width=\textwidth]{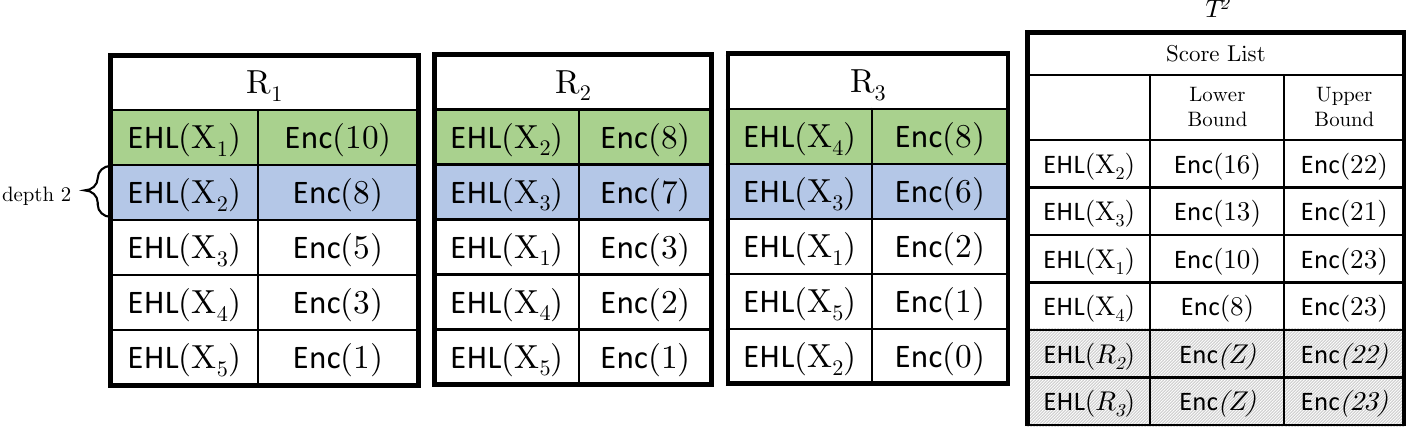}
	        \caption{\centering{$\SecWorst$ $\&$ $\SecBest$ at depth 2. 
	        $T^2$ maintains the sorted encrypted scores based on their worst scores after depth 2. 
	        Note that, after $\SecDedup$, the duplicated objects $X_1, X_2$ do not appear 
	        in the top-$k$ list twice.}}\label{fig:depth2}
    	 \end{subfigure}
   \end{minipage}
   ~
   \begin{minipage}[b]{.6\textwidth}
   	\centering
  	  \begin{subfigure}[t]{\textwidth}
    	   \centering
           \includegraphics[width=1.1\textwidth]{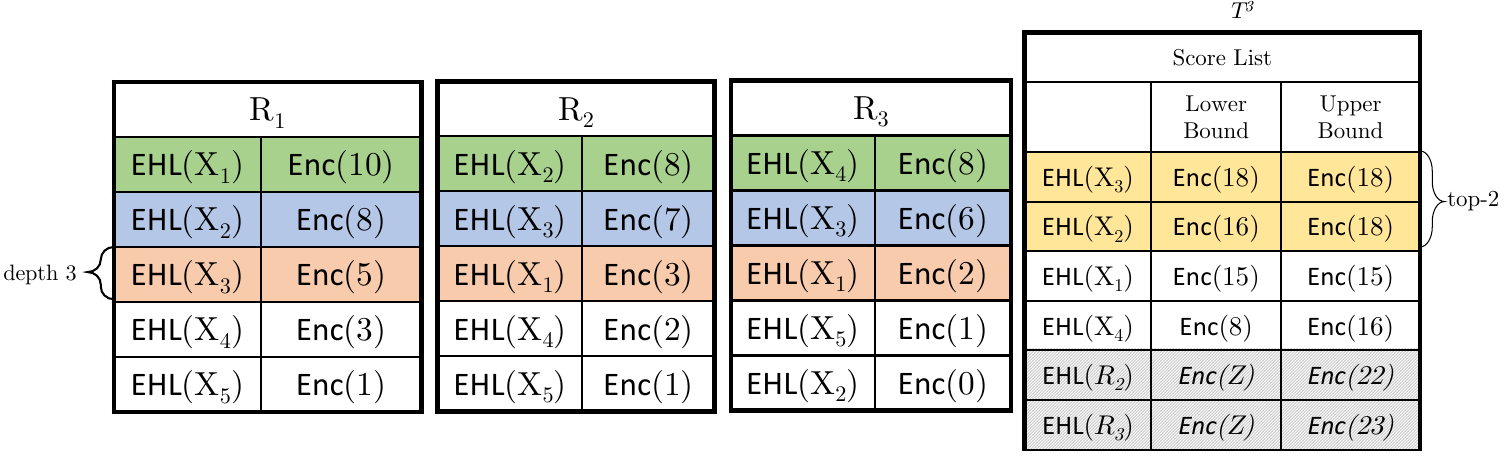}
           \caption{\centering{ $\SecWorst$ $\&$ $\SecBest$ at depth 3. 
           $T^3$ maintains the sorted encrypted scores based on their worst scores after depth 3 . 
           The $\SecQuery$ halts (based on line~\ref{line:halt}) 
           in Algorithm~\ref{alg:topk} }}\label{fig:depth3}
       \end{subfigure}
  \end{minipage}
  \caption{An example of securely computing the top-2 query for $\SecQuery$. The table has three attributes, and the score function $f$ is the sum of all the attributes.}\label{fig:example}
\end{figure}

\subsubsection{Secure Worst Score}\label{subsec:secworst}
At each depth, for each encrypted data item, 
server $S_1$ should obtain the encryption $\En{W}$, which is the worst score based on the items at the current depth \emph{only}.
Note that this is different than the normal NRA algorithm as it computes the global worst possible score for each encountered objects until the current depth.
We formally describe the protocol setup below:
\begin{protocol}\label{protocol:secworst}
Server $S_1$ has the input $E(I)= \angles{\EHL(o),\En{x}}$, a set of encrypted items $H$, i.e. $H= \{E(I_i)\}_{i = [|H|]}$, where $E(I_i) = \angles{\EHL(o_i), \En{x_i}}$, and the public key $\pk_\p$.
Server $S_2$'s inputs are $\pk_\p$ and $\sk_\p$.
$\SecWorst$ securely computes the encrypted worst ranking score based on $L$, i.e., 
$S_1$  outputs $\En{W(o)}$, where $W(o)$ is the worst score based on the list $H$.
\end{protocol}

\begin{example}
(Figure~\ref{fig:depth1}) At depth 1, to compute the worst score (lower bound) for $X_1$, 
$\SecWorst$ takes the encryptions $\Enc(8)$, $\Enc(8)$ at the same depth from columns $R_2$ and $R_3$ 
and finally outputs the $\Enc(10)$ as $10$ is the lower bound for $X_1$ so far after depth 1. 
\end{example}

The technical challenge here is to homomorphically evaluate the
encrypted score only based on the objects' equality relation. 
That is, if the object is the same as another $o$ from $L$, 
then we add the score to $\Enc(W(o))$, otherwise, we don't. 
However, we want to prevent the servers from knowing the relations between the objects at any depth.
We overcome this problem using the protocol $\SecWorst(E(I), L)$ between the two servers 
$S_1$ and $S_2$.  
We present the detailed protocol description of $\SecWorst$ 
in Algorithm~\ref{alg:sec-worst}.

\begin{algorithm}
   \SOne{$E(I)$, $H = \{E(I_j)\}$, $\pk_\p$}
   \STwo{$\pk_\p$, $\sk_\p$}
	
   \ServerOne{
	 Let $|H| = m$. Generate a random permutation $\pi:[m]\to[m]$\;
	 For the set of encrypted items $H = \{E(I_j)\}$, 
	 permute each $E(I_j)$ in $H$ as $E(I_{\pi(j)}) = \EHL(o_{\pi(j)}), \En{x_{\pi(j)}}$.\;
	 
	 \For{each permuted item in $E(I_{\pi(j)})$}
	 {
      	compute $\En{b_j} \larr \EHL(o) \bitminus \EHL(o_{\pi(j)})$,
		send $\En{b_j}$ to $S_2$ \label{line:secworst-homo}
	  
      Receive $\Etwo{t_i}$ from $S_2$ and evaluate:
      	$\Etwo{\En{x_i'}}:=\Etwo{t_i}^{\En{x_i}} 
      	\cdot \Big(\Etwo{1}\Etwo{t_i}^{-1}\Big)^{\En{0}}$\;
      Run $\En{x_i'}\larr\RecoverEnc(\Etwo{\En{x_i'}}, \pk_\p, \sk_\p)$ with $S_2$\;
      }
      Set the worst score $\Enc(W)\larr(\prod_{i = 1}^{m}\En{x_i'})$ \;
	  Output $\Enc(W)$.
	  
	} 
   
   \ServerTwo{
   		\For{each $\En{b_i}$ received from $S_1$}{
   			Decrypt to get $b_i$, set
			$t_i \larr (b_i = 0 \ ? \ 1 : 0)$\;
    		Send $\Etwo{t_i}$ to $S_1$.}\label{line:worst-enc-bit}
	}
\caption{$\SecWorst\big(E(I), H = \{E(I_i)\}_{i\in[|H|]}, \pk_\p, \sk_\p\big)$: Worst Score Protocol}\label{alg:sec-worst}
\end{algorithm}

Intuitively, the idea of $\SecWorst$ is that $S_1$ first generates 
a random permutation $\pi$ and permutes the list of items in $L$.
Then, it computes the $\En{b_i}$ between $E(I)$ and each permuted 
$E(I_{\pi(i)})$, and sends $\En{b_i}$ to $S_2$. 
The random permutation prevents $S_2$ from knowing the pair-wise 
relations between $o$ and the rest of the objects $o_i$'s. 
Then $S_2$ sends  {\small $\Etwo{t_i}$} to $S_1$ (line~\ref{line:worst-enc-bit}). 
Based on Lemma~\ref{lemma:eqi-bit}, $t_i = 1$ if two objects are 
the same, otherwise $t_i = 0$.
$S_1$ then computes
{\small $\Etwo{\En{x_i'}}\larr \Etwo{t_i}^{\En{x_i}} \cdot 
\Big(\Etwo{1}\Etwo{t_i}^{-1}\Big)^{\En{0}}$}.
Based on the properties of DJ Encryption,
\begin{align*}
\Etwo{t_i}^{\En{x_i}} \cdot \Big(\Etwo{1}\Etwo{t_i}^{-1}\Big)^{\En{0}}
= \Etwo{t_i\cdot \En{x_i} + (1-t_i)\cdot \En{0}} = \Etwo{\En{x_i'}}
\end{align*}
Therefore, it follows that $x_i' = 0$ if $t_i = 0$, 
otherwise $x_i' = x_i$.
$S_1$ then runs $\RecoverEnc(\Etwo{\En{x'_i}}, \pk_\p, \sk_\p)$ (describe in Algorithm~\ref{alg:recover-encryption})
to get $\En{x'_i}$. Note that the protocol $\RecoverEnc$ is also used in other protocols.
Finally, $S_1$ evaluates the following equation: 
{\small$\En{W(o)}\larr\prod_{i=1}^{m}\En{x_i'}$}. 
$S_1$ can correctly evaluate the worst score, because that, 
when $t_i = 0$, the object $o_i$ is not the same as $o$, otherwise, 
$t_i =1$. 
The following formula gives the correct computation of the worst score:
\begin{align*}
\prod_{i = 1}^{m} \En{x_i'} & = \En{\sum_{i=1}^{m}{x_i'}}, \text{ where }x_i' =  
\begin{cases} x_i &\mbox{if } o_i = o \\ 
0 & \mbox{otherwise }\end{cases} 
\end{align*}

\begin{algorithm}
\SOne{\Etwo{\En{c}}, $\pk_\p$}   \STwo{$\pk_\p$, $\sk_\p$}

 \ServerOne{
    Generate $r\random \Z_N$,
    compute and send $\Etwo{\En{c+r}} \larr \Etwo{\En{c}}^{\En{r}}$ to $S_2$.
 }
 
 \ServerTwo{Decrypt as $\En{c + r}$ and send back to $S_1$} 

 \ServerOne{
	Receive $\En{c+r}$ and compute: $\En{c} = \En{c+r}\cdot\En{r}^{-1}$\;
	Output $\En{c}$. 
 }

\caption{$\RecoverEnc(\Etwo{\En{c}}, \pk_\p, \sk_\p)$ Recover Encryption}\label{alg:recover-encryption}
\end{algorithm}

Note that nothing has been leaked to $S_1$ at the end of the protocol.  However, there is some leakage function revealed to $S_2$ at current depth, which we will describe it in detail in later section.
However, even by learning this pattern, $S_2$ has still no idea on which particular item is
the same as the other at this depth since $S_1$ randomly 
permutes the item before sending to $S_2$ and everything has been encrypted. 
Moreover, no information has been leaked on the objects' scores.

\subsubsection{Secure Best Score}
The secure computation for the best score is different from computing the worst score. Below we describe the protocol $\SecBest$ between $S_1$ and $S_2$:
\begin{protocol}
Server $S_1$ takes the inputs of the public key $\pk_\p$, $E(I)$ $=$ $\angles{\EHL(o),\En{x}}$ for the object $o$ in list $L_i$, and
 a set of pointers $\P = \{j\}_{i\ne j, j \in \M}$ to the list in $\ER$. Server $S_2$'s inputs are $\pk_\p$, $\sk_\p$. 
The protocol $\SecBest$ securely computes the encrypted best score at the current depth $d$, i.e., $S_1$ finally outputs $\En{B(o)}$, where $B(o)$ is the best score for the $o$ at current depth.
\end{protocol}

\begin{example}
(Figure~\ref{fig:depth2}) At depth 2, to compute the best score (upper bound) for $X_4$, $\SecBest$ 
takes the encryptions seen so far, then based on the scores it outputs $\Enc(23)$ as $23$ is the upper bound for
$X_4$ after depth 2.
\end{example}

\begin{algorithm}[th!]
\SOne{$E(I_i)$ in list $L_i$, $\P= \{j\}_{i \ne j}$, $\pk_\p$}
\STwo{$\pk_\p$, $\sk_\p$}

  \ServerOne{
 	 \ForEach{list $L_i$}{
 	    maintain $\enc(\bottom{x}_{i}^{d})$ for $L_i$, where 
        $\enc(\bottom{x}_i^d)$ is the encrypted score at depth $d$.
     }
 
	 Generate a random permutation $\pi:[l]\to[l]$\;
	 Permute each $L_i$ as $L_{\pi(i)} = \EHL(o_{\pi(i)}), \En{x_{\pi(i)}}$\;
	 
	 \ForEach{permuted $E(I_{\pi(i)})$}{
        compute $\En{b_i} \larr \EHL(o) \bitminus \EHL(o_i)$
	 }
    
     send $\En{b_i}$ to $S_2$ receive $\Etwo{t_i}$  and compute:
      $\Etwo{\En{x_i'}}:=\Etwo{t_i}^{\En{x_i}} 
      	\cdot \Big(\Etwo{1}\Etwo{t_i}^{-1}\Big)^{\En{0}}$\;
	  run $\En{x_i'}\larr\RecoverEnc(\Etwo{\En{x_i'}}, \pk_\p, \sk_\p)$ with $S_2$\;\label{line:secbest-score}  
  
	  compute
	  $\Etwo{\Enc(\bottom{x}_{i}'^{d})}$$\larr$$\big(1$$-$$\prod_{i = 1}^{d}\Etwo{t_i}\big)^{\Enc(\bottom{x}_{i}^{d})}$\;
	  run $\Enc(\bottom{x}_{i}'^{d})$$\larr$$\RecoverEnc(\Etwo{\Enc(\bottom{x}_{i}'^{d})}, \pk_\p, \sk_\p)$ with $S_2$\;
	  set $\Enc(B_i)\larr\Enc(\bottom{x}_{i}'^{d})\cdot(\prod_{i = 1}^{l}\En{x_i'})$\;\label{line:secbest-add}
	  compute $\Enc(B) \larr \prod_{i = 1}^{m}\Enc(B_i)$ and output $\Enc(B)$\;
 }
 \ServerTwo{
   \For{ $\En{b_i}$ received from $S_1$}{
	   Decrypt to get $b_i$.
	    If $b_i = 0$, set $t_i = 1$, otherwise, set $t_i = 0$\;
    Send $\Etwo{t_i}$ to $S_1$.}\label{line:enc-bit}

  }

\caption{$\SecBest\big(E(I_i), \P, \pk_\p, \sk_\p\big)$ Secure Best Score.}\label{alg:sec-best}
\end{algorithm}

At depth $d$, let $E(I)$ be the encrypted item in the list $L_i$, then its best score up to this depth is based on the whether this item has appeared in other lists 
$\{L_j\}_{j \ne i, j \in \M}$.
The detailed description for $\SecBest$ is described in Algorithm~\ref{alg:sec-best}.

In $\SecBest$, $S_1$ has to scan the encrypted items in the other lists to securely evaluate
the current best score for the encrypted $E(I)$. 
The last seen encrypted item in each sorted list contains the encryption of the 
\emph{best possible values} (or \emph{bottom scores}).
If the same object $o$ appears in the previous depth 
then homomorphically adds the object's score to the encrypted best score $\enc(B)$, 
otherwise adds the bottom scores seen so far to $\enc(B)$.
In particular, $S_1$ can homomorphically evaluate (at line~\ref{line:secbest-score}):
$\Etwo{x_i'} = \Etwo{t_{i}\cdot\enc(x_i) + (1-t_{i})\cdot\enc(0)}$. That is, if $t_{i} = 0$ which means 
item $I$ appeared in the previous depth, $x_i'$ will be assigned the corresponding 
score $x_i$, otherwise, $x_i' = 0$. Similarly, $S_1$ homomorphically evaluates the following:
$\enc(\bottom{x}_{i}'^{d})) = \enc((1-\sum_{i}^dt_{i})\cdot\bottom{x}_{i}^{d})$. 
If the item $I$ does not appear in the previous depth, then $(1-\sum_{i}^dt_{i}) = 1$ since
each $t_i = 0$, therefore, $\bottom{x}_{i}'^{d}$ will be assigned to the bottom value
$\bottom{x}_{i}^{d}$. Finally, $S_1$ homomorphically add up all the encrypted scores and 
get the encrypted best scores (line~\ref{line:secbest-add}).

\subsubsection{Secure Deduplication}
At each depth, some of the objects might be repeatedly computed since the same objects may appear in different sorted list at the same depth.
$S_1$ cannot identify duplicates since the items and their scores are probabilistically encrypted.    
We now present a protocol that deduplicates the encrypted objects in the following.

\begin{protocol}
Let the $\vE(I)$ be an encrypted scored item such that $\vE(I) = (\EHL(o), \En{W}, \En{B})$, i.e.
the $\vE(I)$ is associated with $\EHL(o_i)$, its encrypted worst and best score $\En{W_i}$, $\En{B_i}$.
Assuming that $S_1$'s inputs are the public key $\pk_\p$, a set of encrypted scored items $Q = \{\vE(I_i)\}_{i\in[|Q|]}\}$.
Server $S_2$ has the public key $\pk_\p$ and the secret key $\sk_\p$.
The execution of the protocol $\SecDedup$ between $S_1$ and $S_2$ 
enables $S_1$ to get a new list of encrypted distinct objects and their scores, 
that is, at the end of the protocol, $S_1$ outputs a new list of items $\vE(I'_1), ...
,\vE(I'_{l})$, and there does not exist $i, j \in [l]$ with $i\ne j$ such that $o_i = o_j$. 
Moreover, the new encrypted list should not affect the final top-$k$ results.
\end{protocol}
\begin{example}
(Fig~\ref{fig:depth2}) After scanning depth 2, $\SecDedup$ deduplicates the repeated objects in the list $T^2$.
$X_1$ and $X_2$ are the repeated objects. $\SecDedup$ replaces those the objects with random ids $R_1$ and $R_2$ and replaces the worst scores with large number $Z$ so that they do not appear in the top-2 list.
\end{example}

\begin{algorithm}[th!bp]
 \SOne{$\vE(I_1), \dots , \vE(I_{l})$, $\pk_\p$}
 \STwo{$\pk_\p$, $\sk_\p$}
 \SOneOut{Output $\vE(I'_1) \dots \vE(I'_{l})$ without duplicated objects}
 
 \small{
 \ServerOne{
 	Let $|Q| = l$\;
	\For{$i = 1\dots l$}{
    	\For{$j = i+1, \dots, l$}{
	  		Compute 
			$\En{b_{ij}} \larr \big(\EHL(o_i)\bitminus \EHL(o_j)\big)$\;
			\label{line:dep-homo}
	  	}
    	Set the symmetric matrix $\vB$ such that $\vB_{ij} = \En{b_{ij}}$\;
	}
	$S_1$ generate it own public/private key $(\pk', \sk')$\;
	\For{each $\vE(I_{i})$}{
       	  Generate random $\valpha_{i} \in {\Z^k_N}$, $\beta_{i}$, $\gamma_i \in \Z_N$\;
		  Compute $\vE(\Irand_i)=(\EHL(\orand_i),\En{\Wrand_i},\En{\Brand_i})
		  	\larr \cc{Rand}(E(I_i), \valpha_i, \beta_i, \gamma_i)$\;\label{line:rand}
    	  Compute $H_i=\En[\pk']{\valpha_{i}}||\En[\pk']{\beta_{i}}||\En[\pk']{\gamma_{i}}$ using $\pk'$\;
	}
	Generate a random permutation $\pi: [l]\to [l]$\;
    Permute $\pi(\vB)$, i.e. permute $\vB_{\pi(i)\pi(j)}$ for each $B_{ij}$\;
	Permute $\vE(\Irand_{\pi(i)})$ and $H_{\pi(i)}$ for $i \in [1, l]$\;\label{line:permute-1}
	Send $\pi(\vB)$, $\{\vE(\Irand_{\pi(i)})\}^l_{i=1},\{H_{\pi(i)}\}^l_{i=1}$, $\pk'$ to $S_2$\;
   }
  \ServerTwo{
	Receive $\pi(\vB)$, $\{\vE(\Irand_{\pi(i)})\}^l_{i=1}$,
	 $\{H_{\pi(i)}\}^l_{i=1}$, and $\pk'$ from $S_1$\;
   
	\For{upper triangle of $\pi(\vB)$}{
	  decrypt $b_{\pi(i)\pi(j)}:=\dec_{\sk_{\p}}(\vB_{\pi(i)\pi(j)})$\;
      \If{$b_{\pi(i)\pi(j)}$ = 0}{ \label{line:check-dup}
        remove $\vE(\Irand_{\pi(i)}), H_{\pi(i)}$\;  \tcc{Deduplicate items}
        randomly generate $o_i$, and $\valpha_i\in\Z_N^k$,
        $\beta_i,\gamma_i\in\Z_N$\;\label{line:replace-1}
        set $W_i=Z+\beta_i$ and $B_I=Z+\gamma_i$, where $Z = N-1$ \;
        Set $\vE(I'_{\pi(i)}) := (\EHL(o)\odot\Enc(\valpha_i), \Enc(W_i), \Enc(B_i))$\;
        Compute $H'_{\pi(i)}\larr\En[\pk']{\valpha_i}||
			\En[\pk']{\beta_i}||\En[\pk']{\gamma_i}$ using $\pk'$\;\label{line:replace-2}
         }
	 }
     \For{remaining $\En{\Irand_{\pi(j)}}$, $H_{\pi(j)}$}{
        generate random $\valpha'_{i} \in {\Z^k_N}$, $\beta'_{i}$, and $\gamma'_i \in \Z_N$\;
        $\En{I'_i} = (\EHL(o'_i), \En{W'_i}, \En{B'_i})
       	 \larr \cc{Rand}(\Enc(\Irand_{\pi(j)}), \valpha'_i,\beta'_i, \gamma'_i)$\;\label{line:re-rand} 
         $H_{\pi(j)}=\En[\pk']{\valpha_{\pi(j)}},\En[\pk']{\beta_{\pi(j)}},\En[\pk']{\gamma_{\pi(j)}}$\;
         set $H_{i}'=\En[\pk']{\valpha_{\pi(j)}}\cdot\En[\pk']{\valpha_{i}'}||
    		   			\En[\pk']{\beta_{\pi(j)}}\cdot\En[\pk']{\beta_{i}'}||
                  		\En[\pk']{\gamma_{\pi(j)}}\cdot\En[\pk']{\gamma_{i}'}$\;\label{line:rand-homo}  
  	 }
       
  	 Generate a random permutation $\pi':[l]\to[l]$. 
	  Permute new list $\vE(I_{\pi'(i)})$ and $H_{\pi'(i)}'$, 
 	  then send them back to $S_1$\;\label{line:re-permute}
 }
  
 \ServerOne{
      Decrypt each $H_{\pi'(i)}'$ as $\valpha_{\pi'(i)}'$, $\beta_{\pi'(i)}'$, 
       $\gamma_{\pi'(i)}'$ using $\sk'$\;\label{line:rand-decrypt}

  \ForEach{$\vE(I'_{\pi'(i)}) = \EHL(o'_{\pi'(i)}), \En{W'_{\pi'(i)}}, \En{B'_{\pi'(i)}}$}{
     Run and get $\En{\Ihat_i} =(\EHL(\ohat_{i}), \En{\What_i}, \En{\Bhat_i}) \larr
        			\Rand(\Enc(I'_{\pi'(i)}), -\valpha_{\pi'(i)}', -\beta_{\pi'(i)}', 
                    	-\gamma_{\pi'(i)}')$\;\label{line:re-cover}
	}
  Output the encrypted list $\vE(\Ihat_1) ... \vE(\Ihat_{l})$\;
 }
 
 }
\caption{$\SecDedup\big(Q = \{\vE(I_i)\}_{i\in[|Q|]}, \pk_\p, \sk_\p\big)$ : De-duplication Protocol }\label{protocol:dep}
\end{algorithm}

\begin{algorithm}
	\small{
	\caption{$\cc{Rand}\big(\vE(I), \valpha, \beta, \gamma\big)$: Blinding the randomness}
	Let $\vE(I) = (\EHL(o), \En{B}, \En{W})$\;
	 Compute $\vE(\valpha), \Enc(\beta), \Enc(\gamma)$\;
	 Compute
	 $\EHL(o)  \larr \EHL(o)\odot\En{\valpha}$,
	 $\En{W}\larr \En{W}\cdot\En{\beta}$, and
	 $\En{B}\larr \En{B}\cdot\En{\gamma}$\;
	 Output $\vE(I') = (\EHL(o), \En{W}, \En{B})$\;
	 }
\end{algorithm}

\begin{algorithm}[h!]
 \SOne{$\pk_\p$, $T^{d-1}$, $\Gamma^d$ (encrypted list without duplicated objects) }
 \STwo{$\pk_\p$, $\sk_\p$}
 
 \small{
	 \ServerOne{
		Permute $\vE(I_i)\in \Gamma^d$ as $\vE(I_{\pi(i)})$ based on random permutation $\pi$\;

		\ForEach{each permute $\vE(I_{\pi(i)})$}{
			\ForEach{each $\vE(I_j) \in T^{d-1}$}{
				Let $\enc(W_{i})$, $\enc(B_{i})$ be encrypted worst/best score in $\vE(I_{\pi(i)})$
				, and let $\enc(W_{j}), \enc(B_{j})$ be encrypted worst/best score in $\vE(I_{j})$\;
			    Compute $\enc(b_{ij})\larr\EHL(I_{\pi(i)})\bitminus\EHL(I_j)$,
			   send $\enc(b_{ij})$ to $S_2$ and get $\Etwo{t_{ij}}$\;
			    Compute $\Etwo{\Enc(W'_i)}\larr\Etwo{t_{ij}}^{\enc(W_{i})}$,
			    $\enc(W_i')\larr\RecoverEnc(\Etwo{\enc(W'_i)}, \pk_\p, \sk_\p)$, 
			    $\enc(W'_j)\larr\enc(W_j)\enc(W'_i)$\; 
			    Compute $\Etwo{\Enc(B'_j)}\larr\Etwo{t_{ij}}^{\enc(B_{i})}$
			    $\big(\Etwo{1}\Etwo{t_{ij}}^{-1}\big)^{\enc(B_{j})}$
	            $\enc(B'_j)\larr\RecoverEnc(\Etwo{\Enc(B'_j)}, \pk_\p, \sk_\p)$\; 
			   Set $\enc(W'_j),\enc(B'_j)$ as the updated score for $\enc(I_j)$\;
			   compute 
			   $\Etwo{\enc(W'_i)} \larr \Etwo{t_{ij}}^{\enc(W_{i})} \big(\Etwo{1}\Etwo{t_{ij}}^{-1}\big)^{\enc(W'_{j})}$,
			   run $\enc(W'_i)\larr\RecoverEnc(\Etwo{\enc(W'_i)}, \pk_\p, \sk_\p)$ 
			   and maintain $\enc(W'_i)$ for each $E(I_{\pi(i)})$             
			}
			Update the encrypted worst score to $\enc(W'_i)$ 
			for each $\enc(I_{\pi(i)})$ and keep the original best score $\enc(B_i)$\;
		}
	    Append the updated $\enc(I_{\pi(i)})$ to $T^{d-1}$ and get $T^{d}$\;
     	$S_1$ and $S_2$ execute $\SecDedup(T^{d}, \pk_\p, \sk_\p)$ and get the updated list $T^{d}$\;\label{line:update-dedup}
     	$S_1$ finally outputs $T^d$.
  }
 \ServerTwo{
   \ForEach{$\En{b_i}$ received from $S_1$}{Decrypt to get $b_i$\;
	    If $b_i = 0$, set $t_i = 1$, otherwise, set $t_i = 0$.
	    Send $\Etwo{t_i}$ to $S_1$.}\label{line:enc-bit}
 }
 }
\caption{A Secure Update Protocol $\SecUpdate \big(T^{d-1}, \Gamma^d, \pk_\p, \sk_\p\big)$}
\label{protocol:sec-update}
\end{algorithm}

Intuitively, at a high level, $\SecDedup$ let $S_2$ obliviously find the duplicated objects and its scores, and replaces the object id with a random value and its score with a large enough value $Z = N-1\in \Z_N$
(the largest value in the message space) such that, after sorting the worst scores, it will definitely not appear in the top-$k$ list.

\begin{figure}[h]
\centering
 	\includegraphics[width=.6\textwidth]{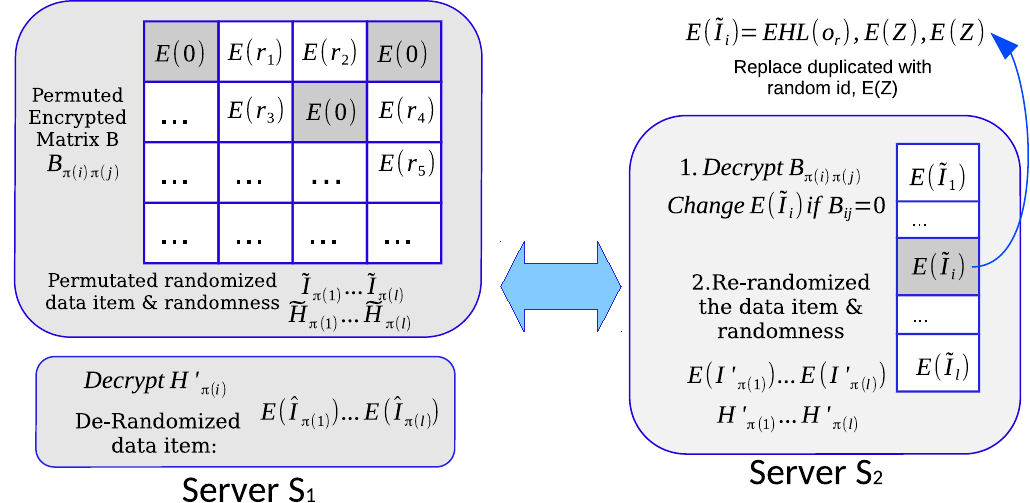}
    \caption{Overview of the $\SecDedup$ protocol}\label{fig:secDedup}
\end{figure}
Figure~\ref{fig:secDedup} gives the overview of our approach.
The technical challenge here is to allow $S_2$ to find the duplicated objects without letting
$S_1$ know which objects have been changed.
The idea is to let the server $S_1$ send a encrypted permuted matrix $\vB$, 
which describes the pairwise equality relations between the objects in the list. 
$S_1$ then use the same permutation to permute the list of blinded 
encrypted items before sending it to $S_2$. 
This prevents $S_2$ from knowing the original data.
For the duplicated objects, $S_2$ replace the scores with a large 
enough encrypted worst score. 
On the other hand, after deduplication, $S_2$ also has to blind the
data items as well to prevent $S_1$ from knowing which items are the duplicated ones.
$S_1$ finally gets the encrypted items without duplication. 
Algorithm~\ref{protocol:dep} describes the detailed protocol.

We briefly discuss the execution of the protocol as follows:
$S_1$ first fill the entry $\vB_{ij}$ by computing $\EHL(o_i)\bitminus\EHL(o_j)$.
Note that, since the encrypted $\vB$ is symmetric matrix indicating the equality
relations for the list, therefore, $S_1$ only need fill the upper triangular for $\vB$.
 and lower triangular can be filled by the fact that $\vB_{ij} = \vB_{ji}$.
In addition, $S_1$ blinds the encrypted item $\Enc(I_i)$
 by homomorphically adding random values and get $\Enc(\Irand_i)$.
This prevents $S_2$ from knowing the values of the item since $S_2$ has the secret key.
Moreover, $S_1$ encrypts the randomnesses using his own public key $pk'$ and get $H_i$.  
To hide the relation pattern between the objects in the list,
$S_1$ applies a random permutation $\pi$ to the matrix 
$\vB_{\pi(i)\pi(j)}$, as well as $\Enc(I_{\pi(i)})$ and $H_{\pi{(i)}}$.  
Receiving the ciphertext, $S_2$ only needs to decrypt
the upper triangular of the matrix, 
$S_2$ only keeps one copy of the $\Enc(\Irand_{\pi(i)})$, $H_{\pi(i)}$ and 
$\Enc(\Irand_{\pi(j)})$, $H_{\pi(j)}$ if $b_{\pi(i)\pi(j)} = 0$. 
Without loss of generality, we keep 
$\Enc(\Irand_{\pi(j)})$, $H_{\pi(j)}$ 
and replace $\Enc(\Irand_{\pi(i)}), H_{\pi(i)}$ as line~\ref{line:replace-1}-\ref{line:replace-2}.
For the unchanged item, $S_2$ blinds them using as well (see line~\ref{line:re-rand}-\ref{line:rand-homo}).
It worth noting that the randomnesses added by $S_2$ are 
to prevent $S_1$ from discovering which item has been changed or not. 
$S_2$ also randomly permute the list as well (line~\ref{line:re-permute}).
$S_1$ homomorphically recovers the original values by decrypting the received 
$H_{\pi'(i)}'$ using his $\sk'$ (see line~\ref{line:re-cover}).
$S_1$ eventually the new permuted list of encrypted items.

For the duplicated objects, the protocol replaces their object id with a random value, and its worst score with a large number $Z$.
For the new encrypted items that $S_2$ replaced (line~\ref{line:replace-1}),
{\small $\En{\Ihat_{i}} = (\EHL(\ohat_{i}), \En{\What_{i}}, \En{\Bhat_{i}})$}, 
we show in the following that {\small $\En{\What_{i}}$} is indeed a new encryption of the permuted {\small $\En{W_{\pi'(\pi(j))}}$} for some $j\in [l]$.
As we can see, the {\small $\En{\What_{i}}$} is permuted by $S_2$'s random $\pi'$, 
i.e. {\small $\En{\What_{\pi'(i)}}$} (see line~\ref{line:re-permute}).
Hence, it follows that:
{\small
\begin{align}
& \En{\What_{\pi'(i)}} \sim \enc\big(W'_{\pi'(i)} - \beta_{\pi'(i)}'\big) \label{eq:1}\\ 
& \sim \enc\big(W'_{\pi'(i)} - (\beta_{\pi'(\pi(j))}+\beta_{\pi'(i)}')\big)\label{eq:2}\\
& \sim \enc\big(\Wrand_{\pi'(\pi(j))} + \beta_{\pi'(\pi(j))} 
       - (\beta_{\pi'(\pi(j))}+\beta_{\pi'(i)}')\big)\label{eq:3}\\
& \sim \enc\big(W_{\pi'(\pi(j))})+ \beta_{\pi'(\pi(j))} + \beta_{\pi'(i)}'
    - (\beta_{\pi'(\pi(j))}+\beta_{\pi'(i)}')\big)\label{eq:4}\\
& \sim \enc\big(W_{\pi'(\pi(j))} \big)
\end{align}
}
In particular, from Algorithm~\ref{protocol:dep}, we can see that Equation~(\ref{eq:1}) 
holds due to line~\ref{line:re-cover}, Equation~(\ref{eq:2}) holds since line~\ref{line:rand-homo} and \ref{line:rand-decrypt}, Equation~(\ref{eq:3})
holds due to line~\ref{line:re-rand}, and Equation~(\ref{eq:4}) holds because of line~\ref{line:rand}.
On the other hand, for the duplicated items that $S_1$ has changed from line~\ref{line:replace-1} to \ref{line:replace-2},
by the homomorphic operations of $S_1$ at line~\ref{line:re-cover}, we have 
{\small
\begin{align*}
\En{\What_{\pi'(k)}} &\sim \enc(W'_{\pi'(k)} - \beta_{\pi'(k)}') 
  \sim \enc(Z + \beta_{\pi'(k)}' - \beta_{\pi'(k)}') \sim \enc(Z)
\end{align*}}
Since $Z$ is a very large enough number, this randomly generated objects definitely do not appear
in the top-$k$ list after sorting.

\subsubsection{Secure Update} 
At each depth $d$, we need to update the current list of objects with the latest global worst/best scores. 
At a high level, $S_1$ has to update the encrypted list $\Gamma^d$ from the state $T^{d-1}$ (previous depth) to $T^{d}$, 
and appends the new encrypted items at this depth. 
Let $\Gamma^d$ be the list of encrypted items with the encrypted worst/best scores $S_1$ get at depth $d$.
Specifically, for each encrypted item $E(I_i)\in T^{d-1}$ and each
$E(I_j)\in L_d$ at depth $d$, we update $I_i$'s worst score by adding the worst from 
$I_j$ and replace its best score with $I_j$'s best score if $I_i = I_j$ since the worst
score for $I_j$ is the in-depth worst score and best score for $I_j$ is the most updated 
best score. If $I_i \ne I_j$, we then simply append $E(I_j)$ with its scores to the list.
Finally, we get the fresh $T^d$ after depth $d$. We describe the $\SecUpdate$
protocol in Algorithm~\ref{protocol:sec-update}.

%

\newcommand{\oput}{\cc{Output}}

\section{Security}\label{sec:security}
Since our construction supports a more complex query type than searching, 
the security has to capture the fact that the adversarial servers also get the 
`views' from the data and meta-data during the query execution. 
The CQA security model in our top-$k$ query processing defines a $\Real$ world and an $\Ideal$ world. In the real world, the protocol between the adversarial servers and the client executes just like 
the real $\SecTopK$ scheme. In the ideal world, we assume that there exists two simulator 
$\SIM_1$ and $\SIM_2$ who get the leakage profiles from an ideal functionality and try to simulate the execution for the real world.
We say the scheme is CQA secure if, after polynomial many queries, no ppt distinguisher can distinguish between the two worlds only with non-negligibly probability. We give the formal security definition in Definition~\ref{def:sec-query}.

\begin{definition}\label{def:sec-query}
Let $\SecTopK = (\enc, \Token, \SecQuery)$ be a top-$k$ query processing scheme 
and consider the following probabilistic experiments where $\E$ is an environment, 
$\C$ is a client, $S_1$ and $S_2$ are two non-colluding semi-honest servers, 
$\SIM_1$ and $\SIM_2$ are two simulators, and $\L_\setup$, 
$\L_\query = (\L^1_\query, \L^2_\query)$ are (stateful) leakage functions:
\begin{description}
\item[\bf $\Ideal(1^\lambda)$:]
   The environment $\E$ outputs a relation $R$ of size $n$ and 
 	sends it to the client $\C$. $\C$ submits the relation $R$ 
    to $\Fideal$, i.e. an \textbf{ideal} top-$k$ functionality.
	$\Fideal$ outputs $\L_\setup(R)$ and $1^\lambda$, and gives 
    $\L_\setup(R)$, $1^\lambda$ to $\SIM_1$.
    Given $\L_\setup(R)$ and $1^\lambda$, $\SIM_1$ generates 
    an encrypted $\ER$. 
    
    $\C$ generates a polynomial number of adaptively chosen queries 
    $(q_1,\ldots,q_m)$. For each $q_i$, $\C$ submits $q_i$ to $\Fideal$, $\Fideal$
    then sends $\L^1_\query(\ER, q_i)$ to $\SIM_1$ and sends $\L^2_\query(\ER, q_i)$
    to $\SIM_2$.
	
	After the execution of the protocol, $\C$ outputs $\out^\Ideal_\C$, 
 	$\SIM_1$ outputs $\out_{\SIM_1}$, and $\SIM_2$ outputs $\out_{\SIM_2}$.

\item[\bf $\Real_{\A}(1^\lambda)$:]
	The environment $\E$ outputs a relation $R$ of size $n$ and 
 	sends it to $\C$. 
    $\C$ computes $(\K, \ER) \larr \enc(1^\lambda, R)$
     and sends the encrypted $\ER$ to $S_1$.
     
    $\C$ generates a polynomial number of adaptively chosen queries $(q_1,\ldots,q_m)$. 
         For each $q_i$, $\C$ computes $\tk_i\larr\Token(K, q_i)$ and sends $\tk_i$
         to $S_1$. $S_1$ run the protocol $\SecQuery\big(\tk_i, \ER\big)$ with $S_2$. 
         
         After the execution of the protocol, $S_1$ sends the encrypted results to $\C$. 
         $\C$ outputs $\out^\Real_\C$, 
 		 $S_1$ outputs $\out_{S_1}$, and $S_2$ outputs $\out_{S_2}$.
\end{description}
We say that $\SecQuery$ is adaptively $(\L_\setup,\L_\query)$-semantically secure (CQA) if the following 
holds:
  \begin{enumerate}[noitemsep, nolistsep]
	\item For all $\E$, for all $S_1$, there exists a ppt simulator $\SIM_1$ 
       such that the following two distribution ensembles are computationally indistinguishable
       \begin{align*}
	       \angles{\out_{S_1}, \out^\Real_\C} \approxeq \angles{\out_{\SIM_1}, \out^\Ideal_\C}
	   \end{align*}
	\item For all $\E$, for all $S_2$, there exists a ppt simulator $\SIM_2$ 
		such that the following two distribution ensembles are computationally indistinguishable
        \begin{align*}
	       \angles{\out_{S_2}, \out^\Real_\C} \approxeq \angles{\out_{\SIM_2}, \out^\Ideal_\C}
	   \end{align*}
   \end{enumerate}
\end{definition}

We formally define the leakage function in $\SecTopK$.
Let the setup leakage $\L_\setup = (|R|, |M|)$, i.e. the size
of the database and the total number of attributes. $\L_\setup$ is the
leakage profile revealed to $S_1$ after the execution of $\Enc$.
During the query processing, we allow $\L_\query = (\L^1_\query, \L^2_\query)$ revealed to the servers.
Note that $\L^1_\query$ is the leakage function for $S_1$, while $\L^2_\query$ 
is the leakage function for $S_2$.
In our scheme, $\L^1_\query = (\QP, D_q)$, where $\QP$ is the \emph{query pattern}
indicating whether a query has been repeated or not.
Formally, for $q_j\in\vq$, the \emph{query pattern $\QP(q_j)$} is a binary vector of length $j$ with a $1$ at location $i$ if $q_j = q_i$ and $0$ otherwise. 
$D_q$ is the halting depth for query $q$.
For any query $q$, we define the equality pattern as follows:
suppose that there are $m$ number of objects at each depth, then
\begin{itemize}
\item \emph{Equality pattern $\EP_d(q)$:} a symmetric binary $m\times m$ matrix $M^d$, 
         where $M^d[i, j] = 1$ if there exist $o_{\pi(i')} = o_{\pi(j')}$ 
         for some random permutation $\pi$ such that $\pi(i')=i$ and $\pi(j')=j$,
         otherwise $M^d[i, j] = 0$.
\end{itemize}
Then, let $\L^2_\query = (\{\EP_d(q)\}^{D_q}_{i = 1})$, i.e. at depth $d\le D_q$ the equality
pattern indicates the number of equalities between objects.
Note that $\EP^d(q)$ does not leak the equality relations between objects at 
any depth in the original database, i.e. the server never knows which objects 
are same since the server doesn't know the permutation.

\newcommand{\HP}{\cc{Hp}}
\newcommand{\Dp}{\cc{Dp}}

\begin{theorem}\label{thm:sec-proof}
Suppose the function used in $\EHL$ is a pseudo-random function and the Paillier encryption is CPA-secure, 
then the scheme $\SecTopK= (\Enc, \Token, \SecQuery)$ we proposed is $(\L_\setup,\L_\query)$-CQA secure.
\end{theorem}

\begin{proof} 
We describe the ideal functionality $\Fideal$ as follows.
An environment $\E$ samples samples $\phi_i\random\bool^{\log n}$ 
for $1 \leq i \leq n$ to create a set $\Phi$ of $n$ distinct identifiers for $n$ objects. For each 
object, $\E$ randomly pick $M$ attribute values from $\Z_N$.
$\E$ finally outputs the relation $R$ with the sampled $n$ objects associated with $M$ attributes.
$\E$ sends the $R$ to the client $\C$, and $\C$ outsources the relation to the ideal functionality
$\Fideal$.
During the setup phase, $\Fideal$ computes $\L_\setup = (n, M)$, 
where $n$ is the number of the objects and $M$ is the number of the attributes. 
Then $\Fideal$ sends $\L_\setup$ to a simulator $\SIM_1$, which we'll describe it next.
$\SIM_1$ receives $\L_\setup$, then $\SIM_1$ generates the random keys 
$\kappa_1, \dots, \kappa_s$ for the $\EHL$ and $\pk_\p, \sk_\p$ for the paillier encryption. 
Then $\SIM_1$ encrypts each object using $\EHL$ and each score using the paillier encryption.  

The client adaptively generate a number of queries $q_1 \dots q_m$.
For each $q_i$, $\Fideal$ computes $\L^1_\query = (\QP, D_{q_i})$ and 
$\L^2_\query = (\{\EP_d(q_i)\}^{D_{q_i}}_{j = 1})$.
$\Fideal$ sends $\L^1_\query$ to $\SIM_1$ and sends
$\L^2_\query = (\{\EP_d(q_i)\}^{D_{q_i}}_{j = 1})$ to
$\SIM_2$.
Next, we describe $\SIM_1$ and $\SIM_2$.
As mentioned above, $\SIM_1$ learns the leakage function $\L^1_\query = (\QP, D_{q_i})$. 
Given the leakage $\QP$, 
$\SIM_1$ first checks if either of the query $q_i$ appeared in any previous query.
If $q_i$ appeared previously, $\SIM_1$ runs the same simulation of $\SecQuery$ as before.
If not, $\SIM_1$ invokes the simulations for each of the sub-routine from $\SecQuery$. 
We show in the Appendix~\ref{app:servers-security} 
(See Definition~\ref{def:sec-servers} and Lemma~\ref{lemma:secworst})
that our building blocks are secure. 
Therefore, $\SIM_1$ can invoke the simulation the original $\SecQuery$ 
by calling the underlying simulators $\SIM_1$ from those sub-protocols 
$\SecWorst$, $\SecBest$, $\SecDedup$, and $\SecUpdate$. As those sub-protocols
have been proved to be secure, the $\SIM_1$ can simulate the execution
of the original $\SecQuery$. Moreover, the messages sent during the execution are
either randomly permuted or protected by the semantic encryption scheme. Therefore,
at the end of the protocol, $\out_{\SIM_1}$ looks indistinguishable from the the output $\out_{S_1}$, 
i.e.
$$\angles{\out_{S_1}, \out_\C} \approxeq \angles{\out_{\SIM_1}, \out_\C'}$$
By knowing $\L^2_\query$, $\SIM_2$ can simulate the execution
of the original $\SecQuery$. Similarly, $\SIM_2$ needs to call the underlying
the simulator $\SIM_2$ from the building blocks $\SecWorst$, $\SecBest$, $\SecDedup$, and $\SecUpdate$. 
As the messages sent during the execution are
either randomly permuted or protected by the semantic encryption scheme, $\SIM_2$ learns nothing
except the leakage $\L_\query^2$. Therefore,
at the end of the protocol, $\out_{\SIM_2}$ looks indistinguishable from the the output $\out_{S_2}$, 
$$\angles{\out_{S_2}, \out_\C} \approxeq \angles{\out_{\SIM_2}, \out_\C'}$$
\end{proof}

\section{Query Optimization}
In this section, we present some optimizations that improve the performance of our protocol.
The optimizations are two-fold: 
1) we optimize the efficiency of the protocol $\SecDedup$ at the expense of 
some additional privacy leakage, and
2) we propose batch processing of $\SecDupElim$ and $\EncSort$ to further 
improve the $\SecQuery$.

\subsection{Efficient SecDupElim}
We now introduce the efficient protocol $\SecDupElim$ that provides similar functionality as $\SecDedup$. 
Recall that, at each depth, $S_1$ runs $\SecDedup$ to deduplicate $m$ encrypted objects, 
 then after the execution of $\SecDedup$ $S_1$ still receives $m$ items but without duplication,
and add these $m$ objects to the list $T^d$ when running $\SecUpdate$. 
Therefore, when we execute the costly sorting algorithm  $\EncSort$ the size of list to sort has $md$ elements at depth $d$.

The idea for $\SecDupElim$ is that instead of keeping the same number encrypted items $m$, 
$\SecDupElim$ {\it eliminates} the duplicated objects. In this way, the number of encrypted
objects gets reduced, especially if there are many duplicated objects.
The $\SecDupElim$ can be obtained by simply changing the $\SecDedup$ as follows: 
in Algorithm~\ref{protocol:dep} at line~\ref{line:check-dup}, 
when $S_2$ observes that there exist duplicated objects, $S_2$ only keeps one copy of them.
The algorithm works exactly the same as before but without performing the 
line~\ref{line:replace-1}-\ref{line:replace-2}. 
We also run $\SecDupElim$ instead of $\SecDedup$ at line~\ref{line:update-dedup} in 
the $\SecUpdate$. That is, after secure update, we only keep the distinct objects with updated
scores. Thus, the number of items to be sorted also decrease.
Now by adapting $\SecDupElim$, if there are many duplicated objects appear in the list,
we have much fewer encrypted items to sort.

\paragraph{Remark on security.}
The $\SecDupElim$ leaks additional information to the server $S_1$.
$S_1$ learns the \emph{uniqueness pattern $\UP^d(q_i)$} at depth $d$, 
where $\UP^d(q_i)$ denotes the number of the unique objects that appear at current depth $d$.
The distinct encrypted values at depth $d$ are independent from all other depths, therefore, 
this protocol still protects the distribution of the original $\ER$. 
In addition, due to the `re-encryptions' during the execution of the protocol, all the encryptions 
are fresh ones, i.e., there are not as the same as the encryptions from $\ER$.
Finally, we emphasize that nothing on the objects and their values 
have been revealed since they are all encrypted. 

\begin{figure}[th!]
\centering
\begin{minipage}{.3\textwidth}
  \centering
  \includegraphics[width=.7\linewidth]{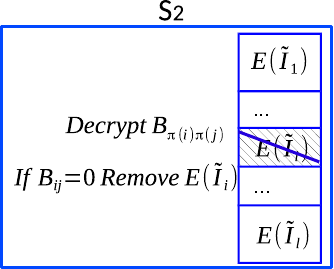}
  \captionof{figure}{$\SecDupElim$}
  \label{fig:test1}
\end{minipage}
 \ \ \ 
\begin{minipage}{.3\textwidth}
  \centering
  \includegraphics[width=\linewidth]{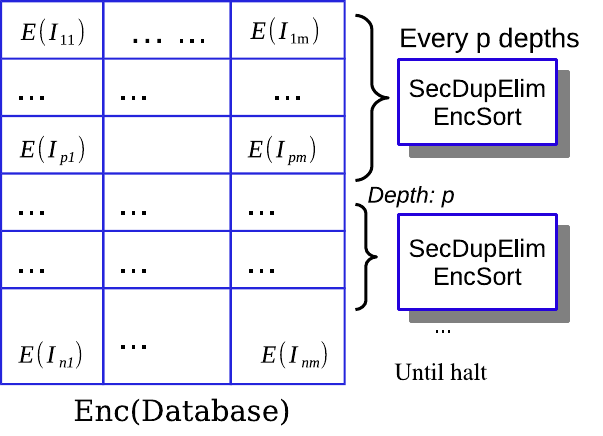}
  \captionof{figure}{Batching Process}
  \label{fig:test2}
\end{minipage}
\end{figure}

\subsection{Batch Processing for SecQuery}
In the query processing $\SecQuery$, we observe that we do not need to run the
protocols $\SecDupElim$ and $\EncSort$  for every depth. 
Since $\SecDupElim$ and $\EncSort$ are the most costly protocols in $\SecQuery$, 
we can perform \emph{batch processing} and execute them after a few depths and not at each depth.
Our observation is that there is no need to deduplicate repeated objects 
at each scanned depth.
If we perform the $\SecDupElim$ after certain depths of scanning, then 
the repeated objects will be eliminated, and those distinct encrypted objects 
with updated worst and best scores will be sorted by running $\EncSort$. The protocol will remain correct.
We introduce a parameter $p$ such that $p \ge k$. The parameter $p$ specifies 
where we need to run the $\SecDupElim$ and $\EncSort$ in the $\SecQuery$ protocol.
That is, the server $S_1$ runs the $\SecQuery$ with $S_2$ the same as in Algorithm~\ref{alg:topk}, except
that every $p$ depths we run line~\ref{line:enc-sort}-\ref{line:halt} in Algorithm~\ref{alg:topk} 
to check if the algorithm could halt. 
In addition, we can replace the $\SecDupElim$ with the original $\SecDedup$ 
in the batch processing for better privacy but at the cost of some efficiency.

\paragraph{Security.}
Compared to the optimization from $\SecDupElim$, we show that the 
batching strategy provides more privacy than just running the $\SecDupElim$ alone.
For query $q$, assuming that we compute the scores over $m$ attributes. 
Recall that the $\UP^p(q)$ at depth $p$ 
has been revealed to $S_1$ while running $\SecDupElim$,
therefore, after the first depth, in the worst case, $S_1$ learns that 
the objects at the first depth is the same object. 
To prevent this worst case leakage, we perform $\SecDupElim$ every $p$ depth.
Then $S_1$ learns there are $p$ distinct objects in the worst case. 
After depth  $p$, the probability that $S_1$ can correctly locate those distinct encrypted 
objects' positions in the table is at most $\frac{1}{(p!)^m}$. 
This decreases fast for bigger $p$. However, in practice this leakage is very small as many distinct
objects appear every $p$ depth.
Similar to all our protocols, the encryptions are fresh due to the `re-encryption' by the server.
Even though $S_1$ has some probability of guessing the distinct objects' location, 
the object id and their scores have not been revealed since they are all encrypted.  

\subsection{Efficiency}\label{complexity}
We analyze the efficiency of query execution. 
Suppose the client chooses $m$ attributes for the query, therefore
at each depth there are $m$ objects. At depth $d$,
it takes $S_1$ $O(m)$ for executing $\SecWorst$, $O(md)$ for executing $\SecBest$, 
$O(m^2)$ for $\SecDedup$,  and $O(m^2d)$ for the $\SecUpdate$. The complexities for $S_2$ are similar.
In addition, the $\EncSort$ has time overhead $O(m\log^2m)$; however,
we can further reduce to $O(\log^2m)$ by adapting parallelism (see~\cite{fc15/FO}). 
On the other hand, the $\SecDupElim$ 
only takes $O(u^2)$, where $u$ is the number of distinct objects at this depth. 
Notice that most of the computations are multiplication (homomorphic addition), 
therefore, the cost of query processing is relatively small.

\newcommand{\syn}{\texttt{synthetic}}
\newcommand{\pamap}{\texttt{PAMAP}}
\newcommand{\diabetes}{\texttt{diabetes}}
\newcommand{\insurance}{\texttt{insurance}}
\newcommand{\qryf}{\cc{Qry\ensuremath{\_}F}}
\newcommand{\qrye}{\cc{Qry\ensuremath{\_}E}}
\newcommand{\qryba}{\cc{Qry\ensuremath{\_}Ba}}
\section{Experiments} \label{experiment}

\begin{figure*}[h!t]
  \begin{minipage}[b]{0.5\textwidth}
  	  \centering
    \begin{subfigure}[t]{.48\textwidth}
        \centering 
        \includegraphics[width=1.09\textwidth]{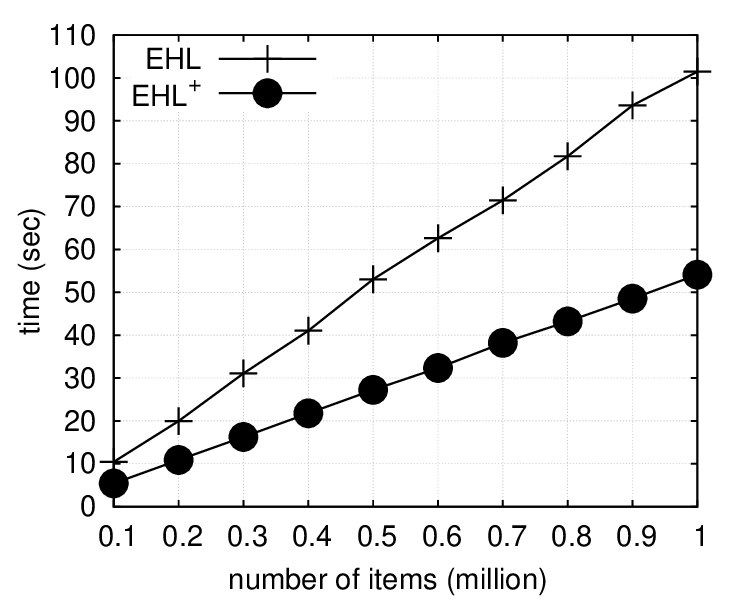}
        \caption{Construction Time}
    \end{subfigure}%
    \begin{subfigure}[t]{.48\textwidth}
        \centering 
        \includegraphics[width=1.09\textwidth]{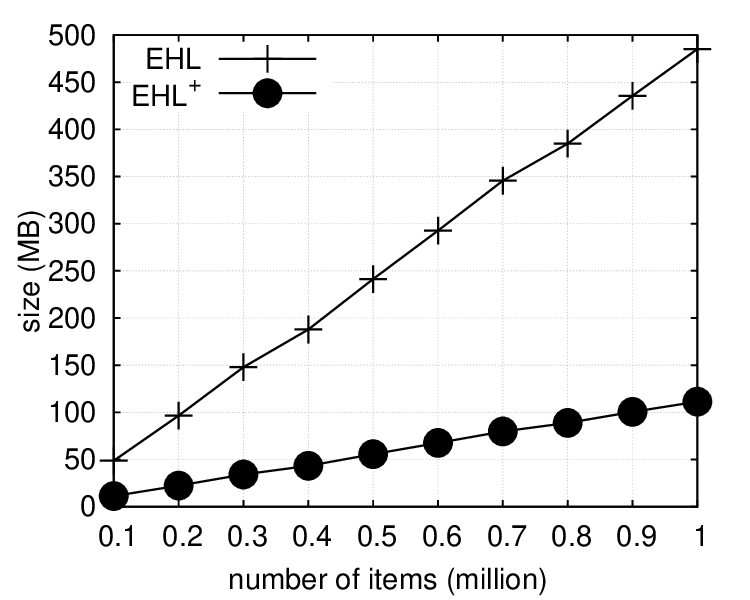}
        \caption{Size Overhead}
    \end{subfigure}
     \caption{Encryption using $\EHL$ vs. $\eEHL$.}\label{fig:construction}
   \end{minipage}
   \quad
   \begin{minipage}[b]{0.5\textwidth}
    	\begin{subfigure}[t]{.48\textwidth}
           \centering 
           \includegraphics[width=1.09\textwidth]{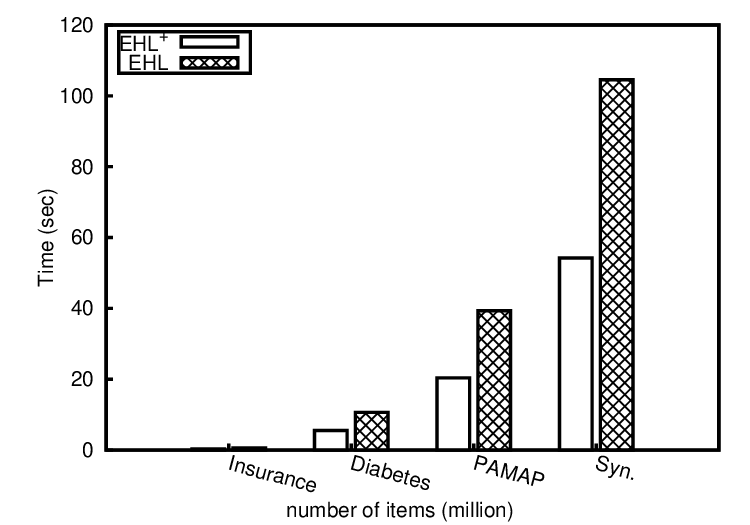}
           \caption{Construction Time}
    	\end{subfigure}
    	\begin{subfigure}[t]{.48\textwidth}
          \centering  
           \includegraphics[width=1.09\textwidth]{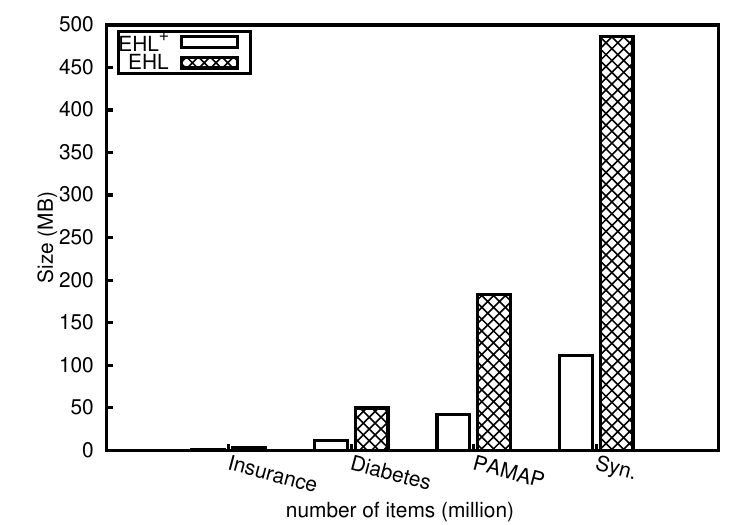}
	   \caption{Size Overhead}
    	\end{subfigure}
          \caption{Encryption $\EHL$ vs. $\EHL^+$ on real data}\label{fig:data-cons}
   \end{minipage}
\end{figure*}
To evaluate the performance of our protocols, we conducted a set of experiments using real and synthetic datasets.
We used the HMAC-SHA-256 as the pseudo-random function (PRF) for the $\EHL$ and $\eEHL$ 
encoding, a 128-bit security for the Pailliar and DJ encryption, and 
all experiments are implemented using C++. We implement the scheme $\SecTopK$ $=$ $(\Enc, \Token, \SecQuery)$,
including all the protocols $\SecWorst$, $\SecBest$, $\EncSort$, and $\EncCompare$ and their optimizations.
We run our experiments on a 24 core machine, who serves as the cloud,
running Scientific Linux with 128GB memory and 2.9GHz Intel Xeon. 
\paragraph{DataSets}
We use the following real world dataset downloaded from UCI Machine Learning Repository~\cite{Lichman:2013}.
$\insurance$: a benchmark dataset that contains $5822$  customers' information on an insurance company and 
we extracted $13$ attributes from the original dataset.
$\diabetes$: a patients' dataset containing $101767$ patients' records 
(i.e. data objects), where we extracted  $10$ attributes.
$\pamap$: a physical activity monitoring dataset
that contains $376416$ objects, and we extracted $15$ attributes.
We also generated synthetic datasets $\syn$ that has 1 million records with $10$ attributes that takes values 
from Gaussian distribution.

\subsection{Evaluation of the Encryption Setup}
We implemented both the $\EHL$ and the efficient $\eEHL$. 
For $\EHL$, to minimize the false positives, we set the parameters as $H = 23$ and $s = 5$, 
where $L$ is the size of the $\EHL$ and $s$ is the number of the secure hash functions.
For $\eEHL$, we choose the number of secure hash function $\hmac$ in $\eEHL$ to be $s = 5$,
and, as discussed in the previous section, we obtained negligible false positive rate in practice. 
The encryption $\Enc$ is independent of the characteristics of the dataset and depends only on the size. 
Thus, we generated datasets such that the number of the objects range from $0.1$ to $1$ million. 
We compare the encryptions using $\EHL$ and $\eEHL$.
After sorting the scores for each attribute, the encryption for each item can be fully parallelized.
Therefore, when encrypting each dataset, we used $64$ threads on the machine that we discussed before.
Figure~\ref{fig:construction} shows that, both in terms of time and space, 
the cost of database encryption $\Enc$ is reasonable and scales linearly to the 
size of the database. Clearly, $\eEHL$ has less time and space overhead.
For example, it only takes $54$ seconds to encrypt $1$ million records
using $\eEHL$. The size is also reasonable, as the encrypted database only takes $111$ MB using $\eEHL$.
Figure~\ref{fig:data-cons} also shows the encryption time and size overhead for the real dataset that
we used. 
Finally, we emphasize that the encryption only incurs a one-time off-line construction overhead.

\subsection{Query Processing Performance}
\subsubsection{Query Performance and Methodology}
We evaluate the performance of the secure query processing 
and their optimizations that we discussed before.
In particular, we use the query algorithm without any optimization but with full privacy, denoted as $\qryf$;
the query algorithm running $\SecDupElim$ instead of $\SecDedup$ at every depth,
denoted as $\qrye$; and the one using the batching strategies, denoted as $\qryba$.
We evaluate the query processing performance using all the datasets and use $\eEHL$ to encrypt all of the object ids.

Notice that the performance of the NRA algorithm depends on the distribution of the dataset among other things.
Therefore, to present a clear and simple comparison of the different methods, we measure the average time per depth for the query processing, i.e. $\frac{T}{D}$, where $T$ is the total time that the program spends on executing a query and $D$ is the total number of depths the program scanned before halting. In most of our experiments the value of $D$ ranges between a few hundred and a few thousands. For each query, we randomly choose the number of attributes $m$ that are used for the ranking function ranging from $2$ to $8$, and we also vary $k$ between $2$ and $20$. The ranking function $F$ that we use is the sum function.

\subsubsection{$\qryf$ evaluation}
We report the query processing performance without any query optimization.
Figure~\ref{fig:full-query} shows $\qryf$ query performance. 
The results are very promising considering that the query is executed completely on encrypted data. 
For a fixed number of attributes $m = 3$, the average time 
is about $1.30$ seconds for the largest dataset $\syn$ running top-$20$ queries.
When fixing $k = 5$, the average time per depth for all the dataset is below $1.20$ seconds. 
As we can see that, for fixed $m$, the performance scales linearly as $k$ increases.
Similarly, the query time also linearly increases as $m$ gets larger for fixed $k$.

\begin{figure*}[h!bp]
  \centering
  \begin{minipage}[b]{0.49\textwidth}
  	  \begin{subfigure}[t]{.48\textwidth}
    	    \centering
            \includegraphics[width=\textwidth]{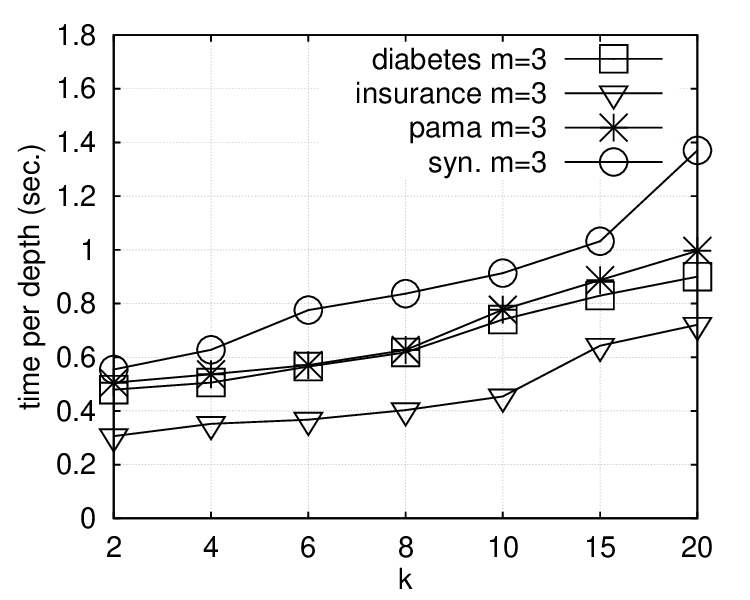}
           \caption{Performance varying $k$}
       \end{subfigure}%
       \begin{subfigure}[t]{.48\textwidth}
   	   \centering  
            \includegraphics[width=\textwidth]{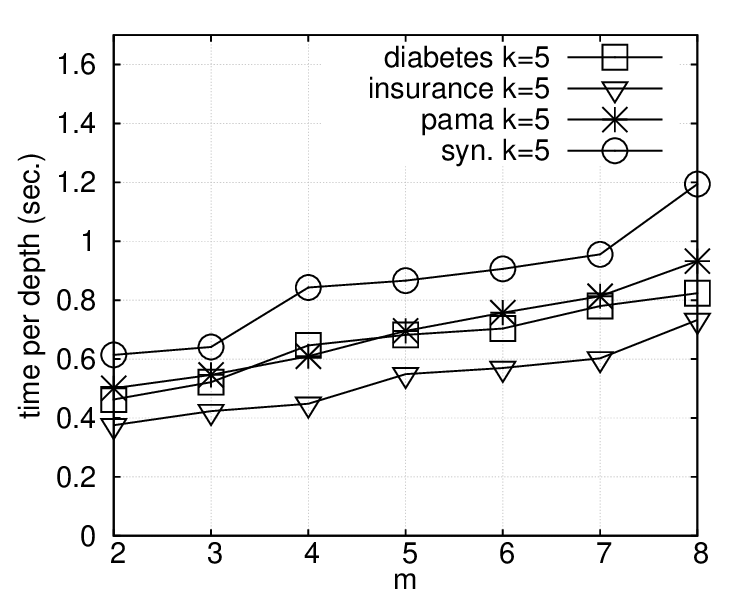}
           \caption{Performance varying $m$}
    	\end{subfigure}
   \caption{$\qryf$ query performance}\label{fig:full-query}
   \end{minipage}
   ~
   \begin{minipage}[b]{0.49\textwidth}
    	\begin{subfigure}[t]{.48\textwidth}
           \centering 
           \includegraphics[width=\textwidth]{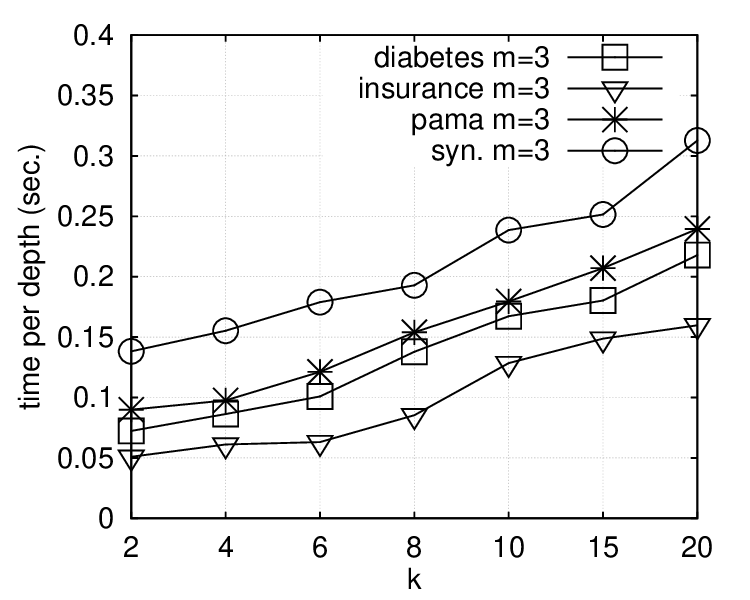}
           \caption{Performance varying $k$}
    	\end{subfigure}
    	\begin{subfigure}[t]{.48\textwidth}
          \centering  
           \includegraphics[width=\textwidth]{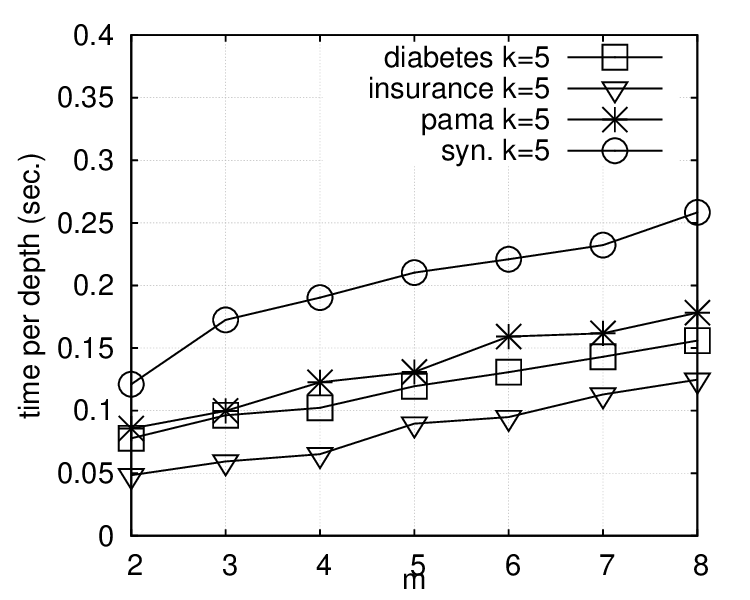}
	   \caption{Performance varying $m$}
    	\end{subfigure}
          \caption{$\qrye$ query optimization performance}\label{fig:elim-query}
   \end{minipage}
\end{figure*}

\subsubsection{$\qrye$ evaluation}
The experiments show that the $\SecDupElim$ improves the efficiency of the query processing.
Figure~\ref{fig:elim-query} shows the querying overhead for exactly the same setting as before.
Since $\qrye$ eliminates all the duplicated the items for each depth, $\qrye$ 
has been improved compared to the $\qryf$ above.
As $k$ increases, the performance for $\qrye$ executes up to $5$ times faster than $\qryf$ when $k$ increase to $20$.
On the other hand, fixing $k = 5$, the performance of $\qrye$ can execute up to around $7$ times faster than $\qryf$ as $m$ grows to $20$.
In general, the experiments show that $\qrye$ effectively speed up the query time 5 to 7 times over the basic approach.

\begin{figure*}[th!bp]
	\centering
  \begin{minipage}[b]{0.9\textwidth}
  	  \begin{subfigure}[t]{0.33\textwidth}
    	    \centering
           \includegraphics[width=\textwidth]{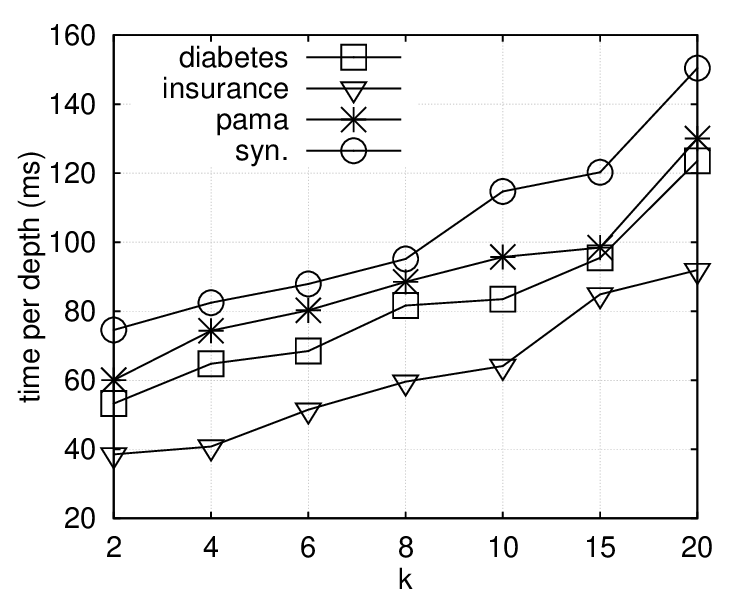}
           \caption{Performance varying $k$}\label{fig:batch-vary-k}
       \end{subfigure}%
       \begin{subfigure}[t]{0.33\textwidth}
	 \centering  
         \includegraphics[width=\textwidth]{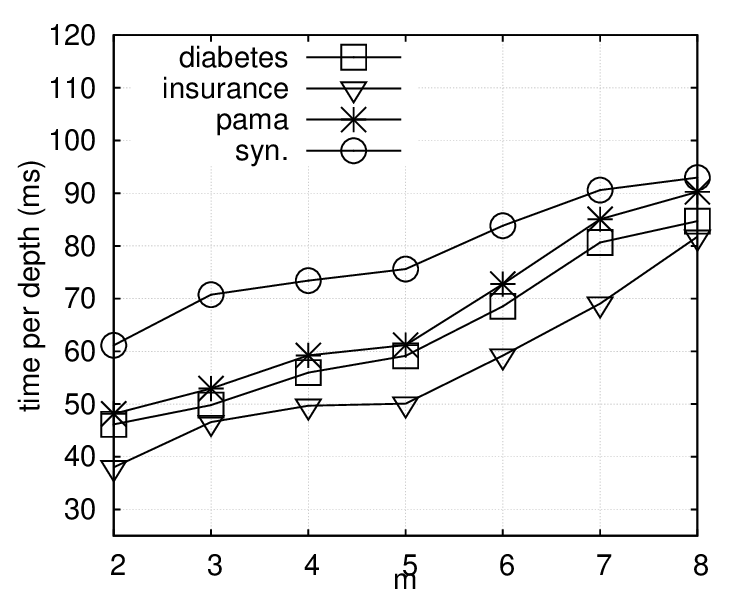}
         \caption{Performance varying $m$}
    	\end{subfigure}
   	\begin{subfigure}[t]{0.33\textwidth}
      	  \centering  
          \includegraphics[width=\textwidth]{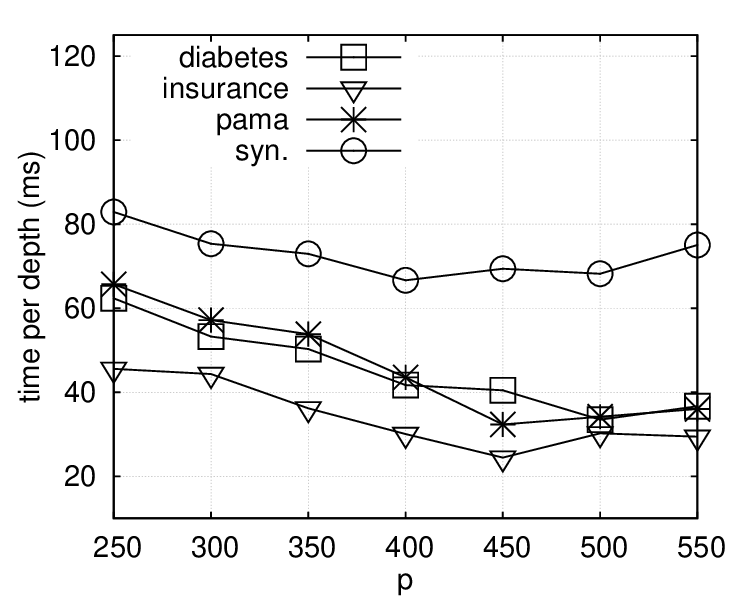}
          \caption{Performance varying $p$}\label{fig:batch-vary-p}
    	\end{subfigure}
      \caption{$\qryba$ query optimization performance}\label{fig:batch}
   \end{minipage}  
\end{figure*}

\begin{figure}[th!]
	\centering
          \includegraphics[width=.3\textwidth]{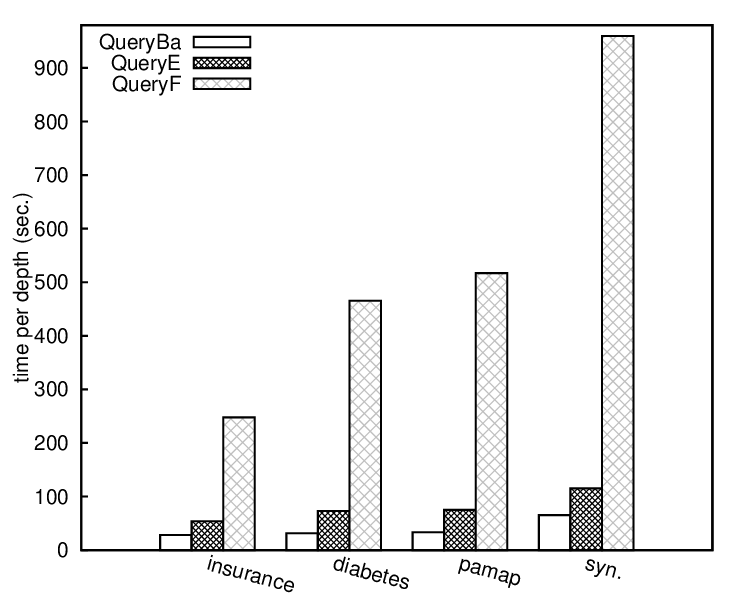}
   \caption{\small Comparisons ($k = 5$, $m = 2$, and $p = 500$)}
   \label{fig:batch-comp}
\end{figure}

\subsubsection{$\qryba$ evaluation}
We evaluate the effectiveness of batching optimization for the $\qryba$ queries.
Figure~\ref{fig:batch} shows the query performance of the $\qryba$ for the same settings as the previous experiments.  
The experiments show that
the batching technique further improves the performance. In particular, for fixed batching 
parameter $p = 150$, i.e. every $150$ depths we perform $\SecDupElim$ and $\EncSort$ in
the $\SecQuery$, and we vary our $k$ from $2$ to $20$. 
Compared to the $\qrye$, the average time per depth for all of the datasets have been further improved. 
For example, when $k = 2$,
the average time for the largest dataset $\syn$ is reduced to $74.5$ milliseconds, 
while for $\qryf$ it takes more than $500$ milliseconds . 
For $\diabetes$, the average time is reduced to $53$ milliseconds when $k = 2$ and $123.5$ milliseconds
when $k$ increases to $20$.
As shown in figure~\ref{fig:batch-vary-k}, the average time linearly increases as $k$ gets larger. 
Similarly, when fixing the $k = 5$ and $p = 150$, for $\syn$ the performance per depth reduce to 
$61.1$ milliseconds and $92.5$ milliseconds when $m = 2$ and $8$ separately.
In Figure~\ref{fig:batch-vary-p}, We further evaluate the parameter $p$. Ranging $p$ from $200$ to $550$, the experiments show that the proper $p$ can be chosen for better query performance. 
For example, the performance for $\diabetes$ achieves the best when $p = 450$. 
In general, for different dataset, there are different $p$'s that can achieve the best query performance. When $p$ gets larger, the number of calls for $\EncSort$ and $\SecDupElim$ are reduced, however, the performance for these two protocols also slow down as there're more encrypted items.

We finally compare the three queries' performance. 
Figure~\ref{fig:batch-comp} shows the query performance when fixing $k = 5$, $m = 3$, and $p = 500$. 
Clearly, as we can see, $\qryba$ significantly improves the performance compared to $\qryf$. 
For example, compared to $\qryf$, the average running time is roughly $15$ times faster for $\pamap$.

\subsubsection{Communication Bandwidth}
We evaluate the communication cost of our protocol. 
\emph{Our experiments show that the network latency is significantly less than the query computation cost.}
In particular, we evaluate the communication of the fully secure and un-optimized 
$\qryf$ queries on the largest  dataset $\syn$. 
For each depth, the bandwidth is the same since the duplicated encrypted 
objects are filled with encryptions of random values.
Each ciphertext is 32 bytes and each round only a few ciphertexts are transferred. Especially,
if we use $m$ attributes in our query we communicate between $m$ and $m^2$ number of ciphertexts each time.
So, for $m=4$, we use between $1024$ to $4096$ bytes messages. Also, the number of messages per depth (per step) is 12.
So, in the worst case, we need to submit 12 total messages between $S_1$ and $S_2$ of 4KB each.
\begin{figure}[h!tbp]
 \centering
  \begin{minipage}[b]{.6\textwidth}
  	  \begin{subfigure}[t]{.5\textwidth}
    	    \centering
           \includegraphics[width=\textwidth]{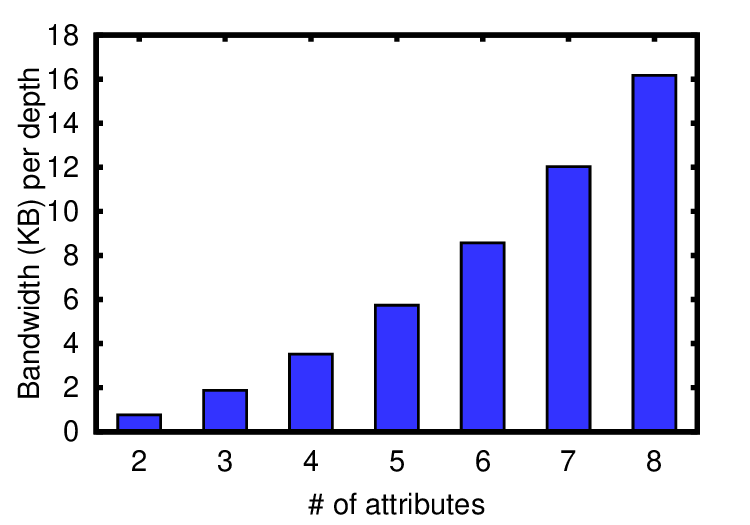}
           \caption{Bandwidth per depth varying $m$}\label{fig:band-m}
       \end{subfigure}
       ~
       \begin{subfigure}[t]{.5\textwidth}
	 	\centering  
         \includegraphics[width=\textwidth]{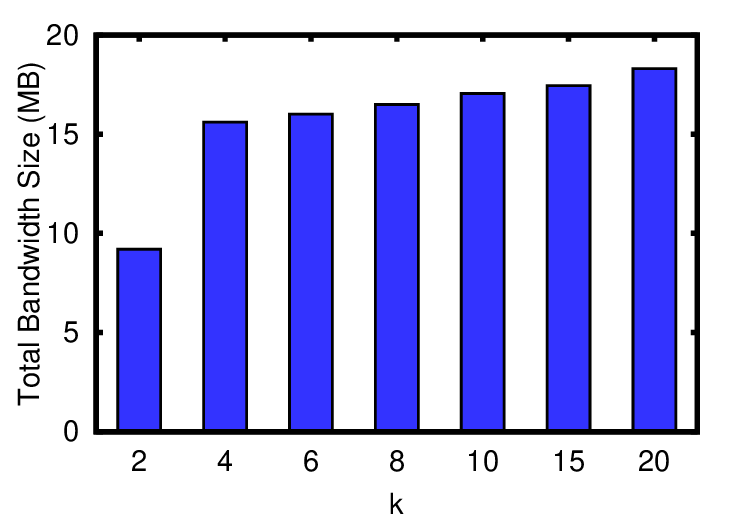}
         \caption{Total bandwidth size for different $k$}\label{fig:band-k}
    	\end{subfigure}
  \end{minipage}
  \caption{Communication bandwidth evaluation}\label{fig:band}
\end{figure}

We evaluate the bandwidth on the largest dataset $\syn$. Note that, the bandwidth per depth is
independent of $k$ since each depth this communication size only depends on $m$. 
As mentioned the bandwidth is $O(m^2)$, by varying $m$
we show in Figure~\ref{fig:band-m} the bandwidth per depth. 
In Figure~\ref{fig:band-k}, we show the total bandwidth when executing the top-$20$ by fixing $m=4$.
As we can see, the total size of the bandwidth is very small, therefore, the total latency could be very small for a high-speed connection between the two clouds.
The speed of the network between two clouds depends on the location and the technology of the clouds.
A recent study showed that we can achieve more than 70 Mbps for 
two clouds where one is in the US and the other in Japan~\cite{DBLP:conf/ndss/Demmler0Z15}.
Furthermore, with recent networking advances~\footnote{\url{http://www.cisco.com/c/en/us/products/cloud-systems-management/intercloud-fabric/index.html}}, we expect that the connections between clouds (inter-clouds) will be much higher~\cite{DBLP:journals/pvldb/BinnigCGKZ16}. 
However, even if we assume that the communication between the two clouds is about $50$ Mbps, 
the total cost of the communication at each depth is below 1 ms!
Thus, communication is not a bottleneck for our protocol. 
In  Figure~\ref{fig:band-m}, we report the actual bandwidth per depth. 
In Figure~\ref{fig:band-k}, we show the total bandwidth when executing the top-$20$ by fixing $m=4$.
As we can see, the total size of the bandwidth is very small that confirms our intuition.
Also, assuming a standard $50$ Mbps LAN setting, we show in Table~\ref{tab:band} the total
network latency between  servers $S_1$ and $S_2$ when $k=20$ and $m=4$.
\begin{table}[t]
\scriptsize
\centering
		\begin{tabular}{|c|c|c|}\hline
		 Dataset & bandwidth (MB) & latency (sec.)\\\hline
		$\insurance$& 8.87  & 1.41 \\\hline
		 $\diabetes$& 12.45 & 1.99 \\\hline
		 $\pamap$   & 15.72 & 2.5152\\\hline
		 $\syn$ 	& 17.3  & 2.768 \\\hline				
		\end{tabular}
	\caption{\scriptsize Comm. bandwidth \& latency ($k=20$, $m=4$)}
		\label{tab:band}
\end{table}
Based on the discussion above, we can see that the communication cost of our protocol is very moderate for any reasonable assumptions about the connectivity between the two clouds. The total cost of our protocol is dominated by the computational cost that was presented in the previous section.

\subsection{Related works on secure $k$NN}
As discussed previously, although existing work~\cite{DBLP:conf/icde/ElmehdwiSJ14} on secure $k$NN does not directly solve our problem, we can adopt their techniques to obtain top-$k$ results by  restricting our scoring function to be $\sum x^2_i(o)$.  We then use as a query a point with large enough values in each attribute and run their secure $k$-nearest-neighbor scheme. In order to compare the experiment from~\cite{DBLP:conf/icde/ElmehdwiSJ14}, during our encryption setup in our $\SecTopK$, the data owner needs to encrypt the additional squares of the values, i.e. $\enc(x^2_i(o))$, then the scoring function would simply be the sum of all the attributes. 
We emphasize that the execution of the rest of our protocol remains the same.

We can now use the results from the experiments in~\cite{DBLP:conf/icde/ElmehdwiSJ14}. It is clear that the protocol in~\cite{DBLP:conf/icde/ElmehdwiSJ14} is very inefficient. For example, it is reported that it would take more than 2 hours to return 10 nearest neighbors for a database of only $2,000$ records. On the other hand, with our scheme, we can return 10 nearest neighbors over a database  with the same characteristics of 1 million records in less than 30 minutes. Moreover, as~\cite{DBLP:conf/icde/ElmehdwiSJ14} needs to sends all of the encrypted records for each query execution, the communication bandwidth is very large even for small dataset that has 2,000 records. On the other hand, in our approach, we show that the bandwidth cost is low and will not affect the performance much.


\section{Top-k Join}\label{app:sec-join}
We would like to briefly mention that our technique can be also extended
to compute top-$k$ join queries over multiple encrypted relations. 
Given a set of relations, $R_1, \dots, R_L$, 
each tuple in $R_i$ is associated with some score that gives it a rank within $R_i$. 
The \emph{top-k join query} joins $R_1$ to $R_L$ and produces the results ranked on a total score.
The total score is computed according to some function, $F$, that combines individual scores. 
We consider only (i.e.\emph{equi-join}) conditions in this paper.
Similarly, the score function $F$ we consider in this paper is also a linear combination 
over the attributes from the joining relations. 
A possible SQL-like join query example is as follows:
\texttt{Q1 = SELECT * FROM A,B,C
WHERE A.1=B.1 and B.2=C.3
ORDER BY A.1+B.2+C.4
STOP AFTER k;
} where \texttt{A, B} and \texttt{C} are three
 relations and \texttt{A.1, B.1, B.2, C.3, C.4} are attributes on these relations. 
Our idea is to design a secure join operator, denoted as $\secjoin$, such that 
the server $S_1$ obliviously joins the relations based on the received token. $S_1$ has to
invoke a protocol with $S_2$ to get the resulting  joined results that meet the 
join condition. 
\subsection{Secure Top-k Join}
We provide a description of the secure top-$k$ join in this section.
Since a join operator is implemented in most system as a dyadic (2-way) operator, 
we describe the secure top-$k$ operator as a \emph{binary} join operator 
between two relations  $R_1$ and $R_2$.
Consider an authorized client that wants to join two encrypted relations and get the top-$k$ based on a join condition.
Assume that each tuple in $R_i$ has $m_i$ many attributes and each $R_i$
have $n_i$ many tuples for $i = \{1,2\}$. Furthermore, denote $o^i_{j}$ be the $j$th objects
in $R_i$ and let $o^i_{j}.x_k$ be the $k$th attribute value. 

\subsection{Encryption Setup for Multiple databases}

\begin{algorithm}[thbp]
 Generate public/secret key $\pk_p, \sk_p$ for the pailliar encryption, generate
     random secret keys $\kappa_1, \dots, \kappa_s$ for the $\EHL$\;
 \ForEach{each $o^i_{j}\in R_i$}{
    \ForEach{each attribute $o^i_{j}.x_k$}{
    	set $E(s_{k})\larr \angles{\EHL(o^i_{j}.x_k), \En{o^i_{j}.x_k}}$\;
	}
	{set $E(o^i_{j}) = \big(E(s_{1}), ..., E(s_{m_i})\big)$}
 }
 Generate a key $\K$ for the PRP $P$\;
 Permutes the encrypted attributes based on $P$, i.e.
	      set $E(o^i_{j}) =  \big(E(s_{P_\K(1)}), ..., E(s_{P_\K(m_i)})\big)$\;
 Output permuted encrypted databases as
  $\ER_1 = \{E(o^1_{1})...E(o^1_{n_1})\}$ and $\ER_2 = \{E(o^2_{1})...E(o^2_{n_2})\}$\;
 \caption{$\Enc(R_1, R_2)$: database encryption}\label{alg:join-encryption}
\end{algorithm}

Consider a set of relations $R_1$ and $R_2$.
The encryption setup is similar as the top-$k$ for one relation. 
The difference is that since we have multiple relations on different data
we cannot assign a global object identifier for each the objects in different relations.
The difference here is that, 
in addition to encrypting an object id with $\EHL$, we encrypt the attribute value using $\EHL$
since the join condition generated from the client is to join the relations based on the
attribute values. 
Therefore, we can compare the equality between different records based on their
attributes. 
The encryption $\Enc(R_1, R_2)$ is given in Algorithm~\ref{alg:join-encryption}.

The encrypted relations  $\ER_1, \ER_2$ do not reveal anything besides the size. 
The proof is similar to the proof in Theorem~\ref{thm:indis}.

\subsection{Query Token}
Consider a client  that wants to run query a SQL-like top-$k$ join as follows:
{\small\sql{Q = SELECT * FROM R1, R2 WHERE R1.A = R2.B ORDER BY R1.C + R1.D STOP AFTER k;} }
\noindent where $\sql{A, C}$ are  attributes in $R_1$ and $\sql{B, D}$ 
are  attributes in $R_2$.
The client first requests the key $\K$ for the $P$, then computes 
{$(t_1,t_2, t_3, t_4)$ $\larr$ $(P_\K(R_1.A), P_\K(R_2.B), P_\K(R_1.C), P_\K(R_2.D))$.}
Finally, the client generates the SQL-like query token as follows:
$t_Q$ = \sql{SELECT * FROM } $\ER_1, \ER_2$ \sql{ WHERE } 
$\ER_1.t_1=\ER_2.t_2$ \sql{ ORDERED BY } $\ER_1.t_3+\ER_2.t_4$ \sql{ STOP AFTER } $k$.
Then, the client sends the token $t_Q$ to the server $S_1$.

\newcommand{\ms}{\tiny}

\subsection{Query Processing for top-k join}\label{sec:join-token}

\begin{algorithm}[h!]

 \SOne{$\pk_\p$, $\tk$}
 \STwo{$\pk_\p$, $\sk_\p$}

	\ServerOne{ 
	
	{Parse $\tk$, let $\JC = (\ER_1.t_1, \ER_2.t_2)$ and $\Score = \ER_1.t_3 + \ER_2.t_4$}
	
		\ForEach{$E(o^{1}_{i}) \in \ER_1$, $E(o^2_{j}) \in \ER_2$ in random order}
		{
		  Let $E(o^{\ms 1}_{i}) = \big(E(x_{i1}), ...E(x_{im_1})\big)$ 
		  and $E(o^2_{j})= \big(E(x_{j1}), ...E(x_{jm_2})\big)$\;
		  Compute $\Enc(b_{ij}) \larr \EHL(x_{it_1})\bitminus\EHL(x_{jt_2})$,
		  where $E(x_{it_1})$, $E(x_{jt_2})$ are the $t_1$-th, 
		  $t_2$-th attributes in $E(o^1_i)$, $E(o^2_j)$\;
		  
		   \tcc{evaluate $s_{ij} = b_{ij}(x_{it_3} + x_{jt_4})$}
		  Send $\Enc(b_{ij})$ to $S_2$ and receive $\Etwo{t_{ij}}$ from $S_2$\;
		  Compute  $\Score$ as: 
		  $s_{ij} \larr \Etwo{b_{ij}}^{\En{x_{it_{\ms 3}}}\En{x_{jt_4}}}$\;
		   run $\En{s_{ij}}\larr \RecoverEnc(S_{ij}, \pk_\p, \sk_\p)$.

		  Combine rest of the attributes for $E(o_{ij})$ as follows:
		  $x_l \larr\Etwo{b_{ij}}^{\En{x_{l}}}$, 
		  where $\En{x_{l}}$ $\in$ $\{\Enc(x_{i1}) ... \Enc(x_{jm_2})\}$
		  in $E(O^1_{i})$, $E(O^2_{j})$\;
		  Run $\En{x_l}\larr \RecoverEnc(x_l, \pk_\p, \sk_\p)$.
		}
	}
	 \ServerTwo{
		{For each received $t_{ij}$, decrypts it. If it is $0$, 
		       then compute $b_{ij}\larr\Etwo{1}$. 
		       Otherwise, $b_{ij}\larr\Etwo{0}$. Sends $b_{ij}$ to $S_1$.}
	}
	 \ServerOne{
	   Finally holds joined encrypted tuples 
	   $E(o_{ij}) = \En{s_{ij}}$, \{$\En{x_l}$\}, 
	   where $\En{s_{ij}}$ is the encrypted $\Score$, 
	   $\{\Enc(x_l)\}$ are the joined attributes from $\ER_1, \ER_2$.
	
	  Run $L \larr\SecFilter(\{E(o_{ij}), \pk_\p, \sk_\p\})$ and get the encrypted list $L$.
	
	  Run $\EncSort$ to conduct encrypted sort on encrypted $\Score$ $\Enc(s_{ij})$, 
	  and return top-$k$ encrypted items.
	}
\caption{$\SecJoin(\tk, \pk_\p, \sk_\p)$: $\secjoin$ with $\JC = (\ER_1.t_1, \ER_2.t_2)$ and $\Score = \ER_1.t_3 + \ER_2.t_4$} \label{alg:join}
\end{algorithm}

In this section, we introduce the secure top-$k$ join operator $\secjoin$. 
We first introduce some notation that we use in the query processing
algorithm.  For a receiving token $t_Q$ that is described  in Section~\ref{sec:join-token},
let the join condition be $\JC \overset{\textbf{def}}{:=} \big(\ER_1.t_1$= $\ER_2.t_2\big)$, and the score function $\Score =
\ER_1.t_3 + \ER_4.t_4$.  Moreover, for each $E(o^1_{i})\in \ER_1$,
let $E(x_{it_1})$ and $E(x_{it_3})$ be the $t_1$-th and
$t_3$-th encrypted attribute. Similarly, let $E(x_{jt_2})$ and
$E(x_{jt_4})$ be the $t_2$-th and $t_4$-th encrypted attribute for each $E(o^2_{j})$ in 
$\ER_2$.
In addition, let $\vE(X)$ be a vector of encryptions, i.e. $\vE(X) = \angles{\En{x_1},..., \En{x_s}}$, and
let $\vE(R) =\angles{(\En{r_1},..., \En{r_s})}$, where $R\in\Z_N^s$ with each $r_i\random\Z_N$ . 
Denote the randomization function
$\Rand$ as below:
\begin{align*}
	\Rand(\vE(X), \vE(R)) & = (\En{x_1}\cdot\En{r_1}, ..., \En{x_n}\cdot\En{r_n})\\
		           & = (\En{x_1+ r_1}, ..., \En{x_n+ r_n})
\end{align*}
. This function is similar to $\Rand$ in Algorithm~\ref{protocol:dep} and
is used to homomorphically blind the original value.

\begin{algorithm}[h!]
\small{
   \SOne{$\{E(o_i)\}$, $\pk_\p$}
   \STwo{$\pk_\p$, $\sk_\p$}

 \ServerOne{
       Let $E(o_i) = \big(\En{s_i}, \vE(X_i)\big)$ where $\vE(X_i) = \angles{\En{x_{i1}}, ...,\En{x_{is}}}$\;
       Generate a key pair $(\pk_s,\sk_s)$\;
       \ForEach{$E(o_i)$}{
	       Generate random $r_i\in \Z_N^*$, and $R_i \in \Z_N^m$\;
	       $\En{s_i'}\larr \En{s_i}^{r_i}$ 
	            and $\vE(X_i')\larr \Rand(\vE(X_i), \vE(R_i))$\label{st.1}\;\label{op:rand}
           Set $E(o'_i) = \big(\En{s_i'}, \vE(X_i')\big)$\;
		}
        Compute the following:
               $\En[\pk_s]{r^{\ms -1}_1}, \En[\pk_s]{R_1}$, $\dots$,
               $\En[\pk_s]{r^{\ms -1}_n}, \En[\pk_s]{R_n}$\;
        Generate random permutation $\pi$, permute $E(o'_{\pi(i)})$, 
            $\En[\pk_s]{r^{-1}_{\pi(i)}}$, and $\En[\pk_s]{R_{\pi(i)}}$\;
        Sends $E(o'_{\pi(i)})$, 
            $\En[\pk_s]{r^{-1}_{\pi(i)}}, \En[\pk_s]{R_{\pi(i)}}$, and $\pk_s$ to $S_2$\;\label{op:per}
   }

 \ServerTwo{
	Receiving the list $E(o'_{\pi(i)})$ and $\En[\pk_s]{r_{\pi(i)}}, \En[\pk_s]{R_{\pi(i)}}$\;
	\ForEach{$E(o'_{\pi(i)})\in\big(\En{s_{\pi(i)}'}, \vE(X_{\pi(i)}')\big)$}{
        decrypt  $b \larr \En{s_{\pi(i)}'}$\;
	    \If{$b = 0$}{
			Remove this entry $E(T'_{\pi(i)})$ and  
            $\En[\pk_s]{r^{-1}_{\pi(i)}}, \En[\pk_s]{R_{\pi(i)}}$
		}
	}
	\ForEach{remaining items}{ 
	      Generate random $\gamma_i\in \Z^*_N$, and $\Gamma_i \in \Z_N^m$\;
		  $\En{\srand_i}\larr\En{s'_{\pi(i)}}^{\gamma_i}$ and 
		  $\vE(\Xrand_i)\larr \Rand\big(\vE(X'_{\pi(i)}), \vE(\Gamma_i)\big)$\;\label{op:re-rand}
          Set $E(\orand_i) = \big(\En{\srand_i}, \vE(\Xrand_i)\big)$\;	   

		 \tcc{\sql{evaluate} $\rrand_i = r^{-1}_{\pi(i)}\gamma^{-1}_i$}

          compute the following using $\pk_s$:
                $\En[\pk_s]{\rrand_i}\larr\En[\pk_s]{r^{-1}_{\pi(i)}}^{\gamma^{-1}_i}$\;                
        \tcc{\sql{evaluate} $\Rrand_i = R_{\pi(i)}+\Gamma^{-1}_i$}

        $\En[\pk_s]{\Rrand_i}\larr$$\En[\pk_s]{R_{\pi(i)}}\cdot\En[\pk_s]{\Gamma_i}$\label{op:add-rand}
	} 
    Sends the $E(\orand_i)$ and $\En[\pk_s]{\rrand_i}, \En[\pk_s]{\Rrand_i}$ to $S_1$
  }  
 \ServerOne{
	\ForEach{$E(\orand_i)= (\En{\srand}, \vE(\Xrand_i))$ and $\En[\pk_s]{\rrand_i}, \En[\pk_s]{\Rrand_i}$}{
		use $\sk_s$ to decrypt $\En{\rrand_i}$ as $\rrand_i$, $\En{\Rrand_i}$ as $\Rrand_i$\;
		\tcc{homomorphically de-blind}
	     compute $\En{\sund_i}\larr\En{\srand_i}^{\rrand_i}$ and 
	     $\vE(\Xund_i)\larr\Rand(\vE(\Xrand_i), \vE(-\Rrand_i))$\;
        Set $E(o'_i) = (\En{\sund_i}, \vE(\Xund_i))$\;\label{op:de-rand}
	} 
	\tcc{Suppose there're $l$ tuples left}
	Output the list $E(o'_1)$ ... $E(o'_{l})$.
 }
}
\caption{$\SecFilter\big(\{E(o_i)\}, \pk_\p, \sk_\p\big)$} \label{alg:sec-filter}
\end{algorithm}

In general, the procedure for query processing includes the following steps:
\begin{itemize}
\item Perform the join on $\ER_1$ and $\ER_2$.
	\begin{itemize}
     \item Receiving the token, $S_1$ runs the protocol with $S_2$ to generate all possible joined 
           tuples from two relations and homomorphically computes the encrypted scores.
     \item After getting all the joined tuples, 
     	$S_1$ runs $\SecFilter(\{E(o_i)\}, \pk_\p, \sk_\p)$ (see Algorithm~\ref{alg:sec-filter}),
	  	which is a protocol with $S_2$ to eliminate the tuples 
        that do not meet the join condition. 
        $S_1$ and $S_2$ then runs the protocol $\SecJoin$.
        $S_1$ finally produce the encrypted join tuples together with their scores.
	\end{itemize}
     \item $\EncSort$: after securely joining all the databases, 
     	$S_1$ then runs the encrypted sorting protocol to get the top-$k$ results.
\end{itemize}

The main $\secjoin$ is fully described in Algorithm~\ref{alg:join}. 
As mentioned earlier, since all the attributes are encrypted, we cannot simply
use the traditional join strategy. The merge-sort or hash based join
cannot be applied here since all the tuples have been encrypted 
by a probabilistic encryption.
Our idea for $S_1$ to securely produce the joined result is as follows: $S_1$ first combines all
the tuples from two databases (say, using nested loop) by initiating the protocol $\SecJoin$.
After that, $S_1$ holds  all the combined tuples together with the scores. 
The joined tuple have $m_1 + m_2$ many attributes (or user selected attributes). Those tuples that
meet the equi-join condition $\JC$ are successfully joined together with the encrypted scores 
that satisfy the $\Score$ function. However, for those tuples that do not meet the 
$\JC$, their encrypted scores are homomorphically computed as $\En{0}$ and their joined attributes are all $\En{0}$
as well.
$S_1$ holds all the possible combined tuples.
Next, the $\SecFilter$  eliminates all of those tuples that do
not satisfy $\JC$. It is easy to see that similar techniques from $\SecDupElim$ can be applied here.
At the end of the protocol, both $S_1$ and $S_2$ only learn
the final number of the joined tuples that meet $\JC$. 


Below we describe the $\SecJoin$ and $\SecFilter$ protocols in detail.
Receiving the $\token$, $S_1$ first parses it as the join condition $\JC$
$=(\ER_1.t_1$, $\ER_2.t_2)$, and the score function $\Score = \ER_1.t_3 +
\ER_2.t_4$.  Then for each encrypted objects $E(o_{1i}) \in \ER_1$ and
$E(o_{2i}) \in \ER_2$ in random order $S_1$ computes $t_{ij}\larr
\big(\EHL(x_{it_1})\bitminus\EHL(x_{jt_2})\big)^{r_{ij}}$, where $x_{it_1}$ and
$x_{jt_2}$ are the value for the $t_1$th and $t_2$th attribute for $E(o_{1i})$
and $E(o_{2j})$ separately.  $r_{ij}$ is randomly generated value in $\Z_N^*$,
then $S_1$ sends $t_{ij}$ to $S_2$.  Having the decryption key, $S_2$ decrypts
it to $b_{ij}$, which indicates whether the encrypted value $x_{it_1}$ and
$x_{jt_2}$ are equal or not. If $b_{ij} = 0$, then we have $x_{it_1}=x_{jt_2}$
which meets the join condition $\JC$. Otherwise, $b_{ij}$ is a random value.
$S_2$ then encrypts the bit $b_ij$ using a double layered encryption and sends
it to $S_1$, where $b_{ij} = 0$ if $x_{it_1}\ne x_{jt_2}$ otherwise $b_{ij} =
1$. Receiving the encryption, $S_1$ computes
$S_{ij} \larr \Etwo{b_{ij}}^{\En{x_{1t_3}}\En{x_{2t_4}}}$, where
$x_{1t_3}$ is the $t_3$-th attribute for $E(o_{1i})$ and 
$x_{2t_4}$ is the $t_4$-th attribute for $E(o_{2j})$.
Finally, $S_1$ runs the $\StripEnc$ to get the normal encryption $\En{s_{ij}}$.
Based on the construction, 
{
\begin{align*}
\En{s_{ij}} \sim \Enc\big(b_{ij} (x_{1t_3}+ x_{2t_4})\big) \text{, where } b_{ij} =
\begin{cases} 
1 &\mbox{if } x_{1t_1} =  x_{2t_2}\\ 
0 & \mbox{otherwise }
\end{cases} 
\end{align*}
}
Finally, after fully combining the encrypted tuples, $S_1$ holds the joined
encrypted tuple as well as the encrypted scores, i.e. $E(T) = (\En{s_{ij}}$,
$\En{x_{11}}...\En{x_{1m_1}}$, $\En{x_{21}}, ..., \En{x_{2m_2}})$.
During the execution above, nothing has been revealed to $S_1$, $S_2$ only
learns the number of tuples meets the join condition $\JC$ but does not which
pairs since the $S_1$ sends out the encrypted values in random order.  Also,
notice that $S_1$ can only select interested attributes from $\ER_1$ and
$\ER_2$ when combining the encrypted tuples. Here we describe the protocol in
general.

After $\SecJoin$, assume $S_1$ holds $n$ combined the tuples with each tuple
has $m$ combined attributes, then for each of tuple $E(T_i) = (\En{s_i},
\vE(X_i))$, where $\vE(X_i)$ is the combined encrypted attributes
$\vE(X_i) = \angles{\En{x_{i1}}, ... \En{x_{im}}}$.  Next, $S_1$ tries to
blind encryptions in order to prevent $S_2$ from knowing the actual value.  
For each $E(T_i)$, $S_1$ generates random $r_i\in \Z_N^*$ and $R_i \in {\Z_N}^m$, and
blinds the encryption by computing following: $\En{s_i'}\larr
\En{s_i}^{r_i}$ and $\vE(X_i')\larr \Rand(\vE(X_i), \vE(R_i))$.  
Then $S_1$ sets $E(T'_i) = \big(\En{s_i'}, \Enc(X_i')\big)$.  
Furthermore, $S_1$ generates a new key pair for the 
paillier encryption scheme $(\pk_s,\sk_s)$ and encrypts the following: 
$L = \En{r^{-1}_1}_{\pk_s}, \En{R_1}_{\pk_s}, \dots, \En{r^{-1}_n}_{\pk_s}, \En{R_n}_{\pk_s}$, 
where each $r_i^{-1}$ is the inverse of $r_i$ in the group $\Z_N$.  
$S_1$ needs to encrypt the randomnesses in order to recover the original values, and we will explain this later.  
Moreover, $S_1$ generates a random permutation $\pi$, then permutes $E(T_{\pi(i)})$ and
$\En{r^{-1}_{\pi(i)}}_{\pk_s}, \En{R_{\pi(i)}}_{\pk_s}$ for $i = [1, n]$.
$S_1$ sends the permuted encryptions to $S_2$.

$S_2$ receives all the encryptions. For each received $E(T_{\pi(i)}') =
(\Enc(s'_{\pi(i)}), \vE(X'_{\pi(i)}))$, $S_2$ decrypts $\Enc(s_{\pi(i)}')$, if
$s_{\pi(i)}'$ is $0$ then $S_2$ removes tuple $E(T_{\pi(i)}')$ and corresponding
$\En{r^{-1}_{\pi(i)}}_{\pk_s}$, $\En{R_{\pi(i)}}_{\pk_s}$.
For the remaining tuples $E(T_{\pi(i)})$,
$S_2$ generates random $r'_i\in \Z^*_N$, and $R'_i \in \Z_N^m$, and compute the following
$\En{\srand_i}\larr\En{s'_{\pi(i)}}^{r'_i}$,
$\Enc(\Xrand_i)\larr \Rand\big(\vE(X'_{\pi(i)}), \vE(R'_i)\big)$ (see Algorithm~\ref{alg:join} line~\ref{op:re-rand}).
Then set $E(\Trand_i) = \big(\En{\srand_i}, \Enc(\Xrand_i)\big)$
Note that, this step prevents the $S_1$ from knowing which tuples have been removed.
Also, in order to let $S_1$ recover the original values,
$S_2$ encrypts and compute the following using $\pk_s$,
$\En{\rrand_i}\larr\En{r^{-1}_{\pi(i)}}_{\pk_s}^{{r'_i}^{-1}}$ 
and 
$\En{\Rrand_i}\larr$$\En{R'_{\pi(i)}}_{\pk_s}\cdot\En{R'_i}_{\pk_s}$.
Finally, $S_1$ sends the $E(\Trand_i)$ and $\En{\rrand_i}, \En{\Rrand_i}$ to $S_1$.
Assuming there're $n'$ joined tuples left. On the other side, $S_1$ receives the encrypted
tuples, for each $E(\Trand_i)= \En{\srand}, \Enc(\Xrand_i)$ and $\En{\rrand_i}, \En{\Rrand_i}$,
$S_1$ recovers the original values by computing the following:
compute $\En{s'_i}\larr\En{\srand_i}^{\rrand_i}$ and $\vE(X'_i)\larr\Rand(\vE(\Xrand_i),\vE(-\Rrand_i))$.
Notice that, for the remaining encrypted tuples and their encrypted scores,
 $S_1$ can successfully recover the original value, we show below that 
the encrypted scores $\En{\sund_j}$ is indeed some permuted $\En{s_{\pi(i)}}$: 

\newcommand{\explain}[1]{{(\small\text{#1})}}
\begin{align*}
\En{\sund_{j}} & \sim \En{\rrand_j\cdot\srand_j} 
		& \explain{see Alg.~\ref{alg:join} line~\ref{op:de-rand}}\\
& \sim \En{r_{\pi(j)}^{-1} r_j'^{-1} \cdot \srand_j} 
		& \explain{see Alg.~\ref{alg:join} line~\ref{op:add-rand}}\\
& \sim \En{r_{\pi(j)}^{-1}r_j'^{-1} \cdot s'_{\pi(j)} \cdot r'_j} 
		& \explain{see Alg.~\ref{alg:join} line~\ref{op:re-rand}}\\
& \sim \En{r_{\pi(j)}^{-1}r_j'^{-1} \cdot s_{\pi(j)} r_{\pi(j)} r'_j} 
		& \explain{see Alg.~\ref{alg:join} line~\ref{op:rand},\ref{op:per}}\\
& \sim \En{s_{\pi(j)}} & \ \
\end{align*}

If we don't want to leak the number of tuples that meet $ $, we can use a similar technique
from $\SecDedup$, that is, $S_2$ generates some random tuples and large enough random scores for 
the tuples to not satisfy $\JC$. In this way, nothing else has been leaked to the servers.
It is worth  noting that the technique sketched above not only can be used for top-$k$ join, but
for any equality join can be applied here.


\subsubsection{Performance Evaluation}
We conduct the experiments under the same environment as in Section~\ref{experiment}. 
We use synthetic datasets to evaluate our sec-join operator $\secjoin$: we uniformly generate $R_1$ with
$5K$ tuples and $10$ attributes, and $R_2$ with $10K$ tuples and $15$ attributes. 
Since the server runs the \emph{oblivious join} that we discuss before over the encrypted databases, 
the performance of the $\secjoin$ does not depend on the parameter $k$. We test the effect of the joined 
attributes in the experiments. We vary the total number of the attribute $m$ joined together from two tables.
Figure~\ref{fig:join} shows performance when $m$ ranges from $5$ to $20$.

\begin{figure}[th!bp]
\centering
 	\includegraphics[width=0.6\textwidth]{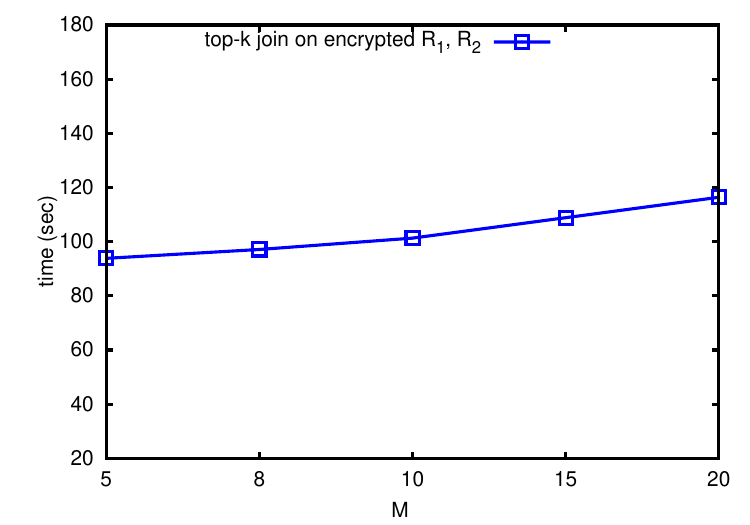}
    \caption{Top-$k$ join: $\secjoin$}\label{fig:join}
\end{figure}

Our operator $\secjoin$ is generically designed for joining any attributes between two relations. 
In practice, one would be only interested in joining two tables using primary-key-to-foreign-key join
or foreign-key-to-primary-key join. Our methods can be easily generalized to those joins. In addition,
one can also pre-sort the attributes to be ranked and save computations in the $\secjoin$ processing. 
We leave this as the future work of this thesis.

\subsection{Related works on Secure Join}
Many works have proposed for executing equi-joins over
encrypted data.
One recent work~\cite{DBLP:journals/tods/PangD14} proposed a 
privacy-preserving join on encrypted data. 
Their work mainly designed for the private join operation, 
therefore cannot support the top-$k$ join. 
In addition, in~\cite{DBLP:journals/tods/PangD14}, although 
the actual values for the joined records are not revealed, 
the server learns some equality pattern on the attributes 
if records are successfully joined.
In addition, \cite{DBLP:journals/tods/PangD14} uses 
bilinear pairing during their query processing, thus it 
might cause high computation overhead for large datasets. 
CryptDB~\cite{DBLP:conf/sosp/PopaRZB11} 
is a well-known system for processing queries on encrypted data. 
MONOMI~\cite{DBLP:journals/pvldb/TuKMZ13} is based on CryptDB with a 
special focus on efficient analytical query processing.
\cite{DBLP:conf/dbsec/KerschbaumHGKSST13} adapts the deterministic 
proxy re-encryption to provide the data confidentiality. 
The approaches using deterministic encryption directly leak the 
duplicates and, as a result, the equality information are leaked
to the adversarial servers.
\cite{DBLP:conf/sigmod/WongKCLY14} propose a secure query system 
SDB that protects the data confidentiality by decomposing the 
sensitive data into shares and can perform secure joins on 
shares of the data. However, it is unclear whether the system can 
perform top-$k$ queries over the shares of the data. 
Other solutions such as Order-preserving encryption 
(OPE)~\cite{BCN11, AKSX04} can also be adapted to secure 
top-$k$ join, however, it is commonly not considered very secure 
on protecting the ranks of the score as the adversarial server 
directly learns the order of the attributes.

\section{Top-k Query Processing Conclusion}
This paper proposes the first complete scheme that executes 
top-$k$ ranking queries over encrypted databases in the cloud. 
First, we describe a secure probabilistic data structure called 
encrypted hash list ($\EHL$) that allows a cloud server to 
homomorphically check equality between two objects without 
learning anything about the original objects.
Then, by adapting the well-known NRA algorithm, we propose 
a number of secure protocols that allow efficient top-$k$ 
queries execution over encrypted data.
The protocols proposed can securely compute the best/worst 
ranking scores and de-duplication of the replicated objects. 
Moreover, the protocols in this paper are stand-alone which 
means the protocols can be used for other applications 
besides the secure top-$k$ query problem.
We also provide a clean and formal security analysis of 
our proposed methods where we explicitly state the leakage 
of various schemes.
The scheme has been formally proved to be secure under the 
CQA security definition and it is experimentally evaluated 
using real-world datasets which show the scheme is 
efficient and practical.

\bibliographystyle{abbrv}
\bibliography{xianrui}

\end{document}